\documentclass[a4paper,onecolumn,oneside,12pt,english]{article}
\usepackage[english]{babel}
\usepackage{amssymb}
\usepackage[dvipsnames]{xcolor}
\usepackage{rotating}
\usepackage{verbatim}
\usepackage[utf8]{inputenc}
\usepackage{tikz}
\usetikzlibrary{positioning,backgrounds,patterns,calc,shapes}
\usepackage{amsmath}
\usepackage{amsthm}
\usepackage{amsfonts,dsfont}
\usepackage{amssymb}
\usepackage{xcolor,graphicx}
\usepackage{tabularx}
\usepackage[numbers,sort&compress]{natbib}
\usepackage{multirow}
\usepackage{multicol}
\usepackage{authblk}
\usepackage{paralist}
\usepackage{a4wide}
\usepackage{enumerate}
\usepackage[pagebackref,menucolor=orange!40!black,filecolor=magenta!40!black,urlcolor=blue!40!black,linkcolor=red!40!black,citecolor=green!40!black,colorlinks]{hyperref} 
\usepackage[sort&compress,nameinlink,noabbrev, capitalize]{cleveref}
\usepackage{paralist}
\usepackage{booktabs}
\usepackage[autostyle]{csquotes}
\usetikzlibrary{decorations.pathreplacing}
\usepackage{scrextend}
\usepackage[subrefformat=parens]{subcaption}
\usepackage[font=small,labelfont=bf]{caption}
\newtheorem{theorem}{Theorem}[section]
\newtheorem{proposition}[theorem]{Proposition}
\newtheorem{lemma}[theorem]{Lemma}
\newtheorem{corollary}[theorem]{Corollary}

\newtheorem{observation}[theorem]{Observation} 

\theoremstyle{definition}
\newtheorem{definition}[theorem]{Definition}

\newtheorem{example}[theorem]{Example}   

\crefname{proposition}{Proposition}{Propositions}
\Crefname{proposition}{Prop.}{Props.}

\crefname{theorem}{Theorem}{Theorems}
\Crefname{theorem}{Thm.}{Thms.}

\crefname{corollary}{Corollary}{Corollaries}
\Crefname{corollary}{Cor.}{Cors.}

\crefname{section}{Section}{Sections}
\Crefname{section}{Sec.}{Secs.}

\newcommand{\problemdef}[3]{
		\begin{center}
	\begin{minipage}{0.95\textwidth}
		\noindent
		\textsc{#1}

				\vspace{2pt}
				\setlength{\tabcolsep}{3pt}
				\begin{tabularx}{\textwidth}{@{}lX@{}}
						\textbf{Input:} 		& #2 \\
						\textbf{Question:} 	& #3
					\end{tabularx}
	\end{minipage}
		\end{center}
}
\makeatletter

\newcommand{\taskdef}[3]{
		\begin{center}
	\begin{minipage}{0.95\textwidth}
		\noindent
		\textsc{#1}

				\vspace{2pt}
				\setlength{\tabcolsep}{3pt}
				\begin{tabularx}{\textwidth}{@{}lX@{}}
						\textbf{Input:} 		& #2 \\
						\textbf{Task:} 	& #3
					\end{tabularx}
	\end{minipage}
		\end{center}
}
\makeatletter
  
\author[1]{Matthias Bentert}
\author[1]{Till Fluschnik\thanks{Supported by the DFG, project DAMM (NI~369/13-2).}${}^{,}$}
\author[1]{André Nichterlein\thanks{Partially supported by a postdoc fellowship of the DAAD while at Durham University, UK.}${}^{,}$}
\author[1]{Rolf~Niedermeier}

\affil[1]{\small Algorithmics and Computational Complexity, Faculty~IV, TU Berlin, Germany, \texttt{\{firstname.lastname\}@tu-berlin.de}}

\tikzstyle{para}=[rectangle,draw=black,minimum height=.8cm,fill=gray!10!white,rounded corners=1mm, on grid] 
\tikzstyle{para2}=[rectangle,draw=black,minimum height=.8cm,fill=gray!10,rounded corners=1mm, on grid, scale=0.7]
\tikzstyle{wnode}=[circle, draw, scale=0.7,fill=white]
\tikzstyle{gnode}=[circle, draw, scale=0.7,fill=lightgray]
\tikzstyle{bnode}=[circle, draw, scale=0.7,fill=black]   

\title{Parameterized Aspects of Triangle Enumeration \footnote{An extended abstract appeared at FCT 2017.}}
\date{\vspace{-1cm}}

\newcommand{\N}{\mathds{N}}
\newcommand{\sol}{\ensuremath{\mathtt{Sol}}}

\newcommand{\hashT}{\operatorname{\#T}}
\newcommand{\hashs}{\operatorname{\#s}}
\newcommand{\hashS}{\operatorname{\#S}}

\newcommand{\NP}{\mathsf{NP}}

\newcommand{\condRef}[2]{{\hyperref[#2]{(#1\ref{#2})}}}

\newcommand*{\defeq}{\mathrel{\vcenter{\baselineskip0.5ex \lineskiplimit0pt
                     \hbox{\scriptsize.}\hbox{\scriptsize.}}}%
                     =}

\newcommand{\trenum}{\textsc{$\triangle$\nobreakdash-Enum}}
\newcommand{\tridet}{\textsc{$\triangle$\nobreakdash-Detect}}
                             
\begin{document}
\title{Parameterized Aspects of Triangle Enumeration\footnote{An extended abstract appeared in the proceedings of \emph{the 21st International Symposium on Fundamentals of Computation Theory~(FCT 2017)} held in Bordeaux, France, September 11--13, 2017, volume 10472 of LNCS, pages 96–110. Springer, 2017.
Several results of the paper were derived in the first authors master thesis~\cite{Bentert16}.}}

\maketitle

\begin{abstract}
The task of listing all triangles in an undirected graph is a fundamental graph
primitive with numerous applications.
It is trivially solvable in time cubic in the number of vertices. 
It has seen a significant body of work contributing to both theoretical aspects (e.g., lower and upper bounds on running time, adaption to new computational models) as well as practical aspects (e.g.\ algorithms tuned for large graphs).
Motivated by the fact that the worst-case running time is cubic, we perform a systematic parameterized complexity study of triangle enumeration.
We provide both positive results (new enumerative kernelizations, ``subcubic'' parameterized solving algorithms) as well as negative results (presumable uselessness in terms of ``faster'' parameterized algorithms of certain parameters such as graph diameter).
To this end, we introduce new and extend previous concepts.
\end{abstract}

\section{Introduction}
Detecting, counting, and enumerating triangles in undirected graphs is a basic graph primitive.
In an~$n$-vertex graph, there can be up to~$n\choose 3$ different triangles and an algorithm checking for each three-vertex subset whether it forms a triangle can list all triangles in~$O({n\choose 3})$ time.
As to counting the number of triangles in a graph, the best known algorithm takes~$O(n^\omega) \subset O(n^{2.373})$ time~\cite{DBLP:journals/tcs/Latapy08} and is based on fast matrix multiplication.\footnote{$\omega$ is a placeholder for the best known $n\times n$-matrix multiplication exponent.}
Consequently, detecting a triangle in a graph is doable in~$O(n^\omega)$ time~\cite{DBLP:journals/tcs/Latapy08} and it is conjectured that every algorithm for detecting a triangle in a graph takes at least~$\Omega(n^{\omega-o(1)})$ time~\cite{DBLP:conf/focs/AbboudW14}.
We mention that for $m$-edge graphs there is also an~$O(m^{1.5})$-time algorithm \cite{DBLP:journals/siamcomp/ItaiR78} which is interesting in case of sparse graphs.
Our work is motivated by trying to break such (relative or conjectured) lower bounds and improve on best known upper bounds---the twist is 
to introduce
a secondary measurement beyond mere input size. 
This is also known as problem parameterization. 
While parameterizing problems with the goal 
to achieve fixed-parameter tractability results  is 
a well-established line 
of research for $\NP$-hard problems, \emph{systematically} applying and extending 
tools and concepts 
from parameterized algorithmics to 
polynomial-time solvable problems is still in its 
infancy~\cite{BDKNN,DBLP:conf/soda/AbboudWW16,Husfeldt16,FKMNNT2017,GMN17,MertziosNN17,CoudertDP18,FominLSPW18}. 
Performing a closer study of mostly 
triangle enumeration, we contribute to this line of research, also 
referred to as ``FPT~in~P'' for short~\cite{GMN17}. 
Our central leitmotif herein is the quest for parameterized subcubic triangle 
enumeration algorithms.\footnote{Note that there is also the concept of P-FPT where the running time may only have a polynomial dependency on the parameter. We use the more lenient concept of FPT~in~P since there are known cases where an exponential dependency on the parameter is ``needed'' for a linear dependency on the input size. }

\paragraph{Related work}
Triangle enumeration, together with its relatives counting and detection, has many applications, ranging from spam detection~\cite{BecchettiBCG10} over complex network analysis~\cite{GGKT15,DBLP:journals/siamrev/Newman03} and database applications \cite{DBLP:journals/tods/KhamisNRR16} to applications in bioinformatics~\cite{DBLP:conf/icde/ZhangP12}. 
We refer to Kolountzakis et.\,al.\;\cite{DBLP:journals/im/KolountzakisMPT12} for an extended list of applications with small explanations why the respective triangle detection problem is relevant in each application.

The significance of the problem triggered substantial theoretical and practical work.
The theoretically fastest algorithms are based on matrix multiplication and run in~$O(n^\omega + n^{3(\omega-1)/(5-\omega)} \cdot \allowbreak \hashT^{2(3-\omega)/(5-\omega)})$ time, where~$\hashT$ denotes the number of listed triangles~\cite{DBLP:conf/icalp/BjorklundPWZ14}.
Furthermore, there is work (including heuristics and experiments) on listing triangles in large graphs~\cite{DBLP:journals/datamine/LagraaS16,DBLP:conf/wea/SchankW05}, on triangle enumeration in the context of map reduce in graph streams~\cite{DBLP:conf/cikm/ParkSKP14,DBLP:conf/soda/Bar-YossefKS02,DBLP:journals/algorithmica/BulteauFKP16}, and even on quantum algorithms for triangle detection~\cite{DBLP:journals/algorithmica/LeeMS17}. 
For a broader overview, we refer to a survey by Latapy~\cite{DBLP:journals/tcs/Latapy08}.

As to parameterized results, early work by Chiba and Nishizeki~\cite{DBLP:journals/siamcomp/ChibaN85} showed that all triangles in a graph can be counted in~$O(m\cdot d)$ time, where~$d$ is the degeneracy of the graph.\footnote{A graph~$G$ is $d$-degenerate if every subgraph of~$G$ contains a vertex of degree at most~$d$. The degeneracy of~$G$ is the smallest~$d$ such that~$G$ is~$d$-degenerate.
Thus $G$ contains at most~$n\cdot d$ edges. 
Indeed, Chiba and Nishizeki~\cite{DBLP:journals/siamcomp/ChibaN85} used the parameter ``arboricity'' of a graph, which is at most two times its degeneracy. Moreover, their algorithm can be modified to list all triangles without any substantial overhead.}
This running time can be improved by saving polylogarithmic factors if the number of triangles is not cubic in the number of vertices in the input graph~\cite{DBLP:conf/wads/KopelowitzPP15}, but the 3SUM-conjecture\footnote{The 3SUM problem asks whether a given set~$S$ of~$n$ integers contains three integers $a, b, c \in S$ summing up to 0. The 3SUM-conjecture~\cite{GajentaanO95} states that for any constant~$\varepsilon > 0$ there is no~$O(n^{2-\varepsilon})$-time algorithm solving 3SUM. The connection between 3SUM and listing/detecting triangles is well studied~\cite{DBLP:journals/algorithmica/LeeMS17,DBLP:conf/stoc/Patrascu10}.}
rules out more substantial improvements~\cite{DBLP:conf/soda/KopelowitzPP16}.
Green and Bader~\cite{DBLP:conf/socialcom/GreenB13} described an algorithm for triangle counting running in~$O(T_K + |K| \cdot \Delta_K^2)$~time, where~$K$ is a vertex cover of the input graph, $\Delta_K$ is the maximum degree of vertices in $K$ (with respect to the input graph), and $T_K$ is the time needed to compute $K$. 
They also described several experimental results.
Recently, \citet{CoudertDP18} proved that both, detecting and counting triangles in a given graph, can be solved in~$O(k^2\cdot(n+m))$ time, where~$k$ is (amongst others) the clique-width of the input graph.
Finally, Eppstein and Spiro \cite{EppsteinS12} described a data structure that keeps track of the number of triangles in an undirected graph that changes over time (in each step a new edge can be inserted or an existing one can be deleted). They described the time complexity of updating their data structure in terms of the~$h$-index of the current graph.\footnote{The~$h$-index of a graph is the maximum number~$h$ such that the graph contains at least~$h$ vertices of degree at least~$h$.}

Complementing algorithmic results, the problems of detecting and counting triangles have been studied from a running-time-lower-bounds perspective, where the lower bounds are based on popular conjectures like, for instance, the Strong Exponential-Time Hypothesis~\cite{IPZ01,IP01}.
Accordingly, lower bounds have been proven for detecting (in a given edge-weighted graph) a triangle of negative weight~\cite{WilliamsW18} or counting (in a given vertex-colored graph) triangles with pairwise different colors on the vertices~\cite{AbboudWY18}.
Recently, making use of the above mentioned results~\cite{WilliamsW18,AbboudWY18}, lower bounds on \emph{strict} kernelization, that is kernelization where the parameter in the kernel is not allowed to increase, have been proven for both problems when parameterized by, for instance, the maximum degree or the degeneracy of the given graph~\cite{FluschnikMN18};
Additionally, for both problems with the same parameterization, a strict kernel of cubic size computable in~$O(n^{5/3})$~time has been shown.

\paragraph{Our contributions}
We systematically explore the parameter space for triangle enumeration and classify the usefulness of the parameters for FPT-in-P algorithms.
In doing so, we extend a concept of enumerative kernelization given by Creignou et al.~\cite{DBLP:conf/mfcs/CreignouMMSV13} and present a novel hardness concept, as well as algorithmic results. 
Our concrete results are surveyed in \cref{fig:resulttab}. We refer to the respective sections for motivation and a formal definition of the various parameters.
\begin{table}[t!]
\caption{Overview of our results ($n$: number of vertices; $m$: number of edges; $\hashT{}$:~number of triangles; $k$: respective parameter; $\Delta$: maximum degree).}
\label{fig:resulttab}
	\setlength{\tabcolsep}{10pt}
	\renewcommand{\arraystretch}{1.15}
	\begin{tabularx}{\textwidth}{llll}
		\toprule
		& parameter~$k$ & result & reference\\
		\midrule
		\multirow{4}{*}{\rotatebox[origin=c]{90}{enum-kernel}}
				& \multirow{2}{*}{feedback edge number} & \multirow{2}{*}{\LARGE{\{}} size at most 9$k$ & \multirow{2}{*}{\Cref{fesnkernel}}\\
				& &\ \ \ \ in~$O(n+m)$ time  & \\
				& \multirow{2}{*}{distance to~$d$-degenerate} & \multirow{2}{*}{\LARGE{\{}} at most~$k+2^{k}+3$ vertices & \multirow{2}{*}{\Cref{thm:dtddkernel}}\\
				& &\ \ \ \ in~$O(n \cdot d \cdot (k + 2^{k}))$ time  & \\
		\midrule
		\multirow{8}{*}{\rotatebox[origin=c]{90}{solving}}
				& distance to~$d$-degenerate & \multirow{2}{*}{$O(k\cdot \Delta^2 + n \cdot d^2)$} & \multirow{2}{*}{\Cref{Bader}} \vspace{-2.5pt}\\
				& \hfill + maximum degree~$\Delta$ &  & \\
				& feedback edge number & $O(k^2+n+m)$ & \Cref{cor:fesnalg}\\
				& distance to~$d$-degenerate & $O(n \cdot d \cdot (k+d) + 2^{3k} + \hashT)$& \,\,\,\Cref{cor:dtddalg}\\
				& distance to bipartite &$O(\hashT{} + n + m \cdot k)$& \Cref{thm:distbipchorcog}\ref{thm:distbip}\\
				& distance to chordal &$O(n + m \cdot k)$& \Cref{thm:distbipchorcog}\ref{thm:distchor}\\
				& distance to cographs &$O(\hashT{} + n + m \cdot k)$& \Cref{thm:distbipchorcog}\ref{thm:distcog}\\
				& clique-width &$O(n^2 + n \cdot k^2 + \hashT{})$ & \Cref{thm:cliquewidth}\\
		\midrule
		\multirow{3}{*}{\rotatebox[origin=c]{90}{hardness}}
				& domination number, & for~$k \geq 3$ & \multirow{3}{*}{\Cref{chromnum}}\vspace{-0pt}\\
				& chromatic number, and & as hard as & \vspace{-2.5pt}\\
				& diameter & the general case & \\
		\bottomrule
	\end{tabularx}
\end{table}

In particular, we provide \emph{enumerative problem kernels} with respect to the parameters ``feedback edge number'' and ``distance to $d$-degenerate graphs''.
Partially based on data reduction algorithms, we provide fast algorithms for several parameters such as feedback edge number, (vertex-deletion) distance to bipartite graphs, chordal graphs, cographs and to $d$-degeneracy (the last one with and without the additional parameter maximum vertex degree), and clique-width.
On the negative side, using a concept we call ``\emph{General-Problem-hardness}'', we show that using the parameters domination number, chromatic number, and diameter do not help to get FPT-in-P algorithms for detecting triangles, that is, even for constant parameter values the problem remains as ``hard'' as the general version with unbounded parameter.

\paragraph{Organization}
The remainder of this work is organized as follows. In \cref{sec:Prelims} we fix some notation and explain basic concepts. Then, in \cref{sec:concepts}, we 
introduce a new notion of hardness (GP-hardness) and of kernelization (enum-advice kernelization), give a small example for each of them, and state our main negative result. In \cref{sec:algs} we provide the positive results of our work. We conclude in \cref{sec:conclusion} and state some directions for further research.
\section{Preliminaries}
\label{sec:Prelims}
\paragraph{Notation}
For an integer~$\ell \ge 1$, let~$[\ell] = \{1,\ldots,\ell\}$.
Let~$G = (V, E)$ be an undirected simple graph. 
We also denote by~$V(G)$ and $E(G)$ the vertex set and the edge set of~$G$, respectively.
We set~$n \defeq |V|$,~$m \defeq |E|$, and~$|G| \defeq n + m$. 

We denote by~$N(v)$ the (open) neighborhood of a vertex~$v\in V$ and by~${\operatorname{deg}(v) \defeq |N(v)|}$ the degree of~$v$. 
By~$G[U]$ we denote the subgraph of~$G$ induced by the vertex subset~$U \subseteq V$ and~$G - U \defeq G[V \setminus U]$.
A vertex subset~$U\subseteq V$ is a vertex cover of~$G$ if~$G-U$ contains no edges.
If~$\{x,y,z\}\subseteq V$ induces a triangle in a graph, we refer to~$T=\{x,y,z\}$ as this triangle.
We denote the number of triangles in the graph by~$\hashT$. 
Our central problem is as follows.
\taskdef{Triangle Enumeration (\trenum)}
        {An undirected graph~$G$.}
        {List all triangles contained in~$G$.}
\paragraph{Parameterized Complexity}
A language~$L\subseteq \Sigma^* \times \N$ is a \emph{parameterized problem} over some finite alphabet~$\Sigma$, where~$(x,k) \in \Sigma^* \times \N$ denotes an instance of~$L$ and~$k$ is the parameter.
For a parameterized problem~$L$, the language~$\hat{L}=\{x\in \Sigma^*\mid \exists k\colon (x,k)\in L\}$ is called the \emph{unparameterized problem} associated with~$L$.
Then~$L$~is called \emph{fixed-parameter tractable} (equivalently, $L$~is in the class FPT) if there is an algorithm that on input~$(x,k)$ decides whether~$(x,k)\in L$ in~$f(k)\cdot |x|^{O(1)}$ time, where~$f$ is some computable function only depending on~$k$ and~$|x|$ denotes the size of~$x$.
We call an algorithm with a running time of the form~$f(k)\cdot |x|$ a linear-time FPT algorithm.
Creignou et al.~\cite[Definition 3.2]{DBLP:conf/mfcs/CreignouMMSV13} introduced the concept of FPT-delay algorithms for enumeration problems. 
An algorithm~$\mathcal{A}$ is an \emph{FPT-delay algorith}m if there exist a computable function $f$ and a polynomial $p$ such that $\mathcal{A}$~outputs for every input~$x$ all solutions of~$x$ with at most~$f(k) \cdot p(|x|)$ time between two successive solutions.
If the delay can be upper-bounded in~$p(|x|)$, then the algorithm is called a~\emph{$p$-delay algorithm}.
A \emph{kernelization} for~$L$ is an algorithm that on input~$(x,k)$ computes in time polynomial in~$|x|+k$ an output~$(x',k')$ (the \emph{kernel}) such that

\begin{enumerate}[(i)]
	\item $(x,k)\in L \iff (x',k')\in L$, and
	\item $|x'|+k'\leq g(k)$ for some computable function~$g$ only depending on~$k$.
\end{enumerate}

The value~$g(k)$ denotes the \emph{size of the kernel}.
It is well-known that a decidable parameterized problem~$P$ is in FPT if and only if~$L$~admits a kernelization.

We can associate a graph parameter~$k$ with a function~$\kappa\colon\mathcal{G}\to \N\cup\{\infty\}$ that maps every graph~$G$ to its parameter value, where~$\mathcal{G}$ denotes the family of all graphs.
We say that~$k$ \emph{lower-bounds} a parameter~$k'$, associated with function~$\kappa'$, if there is a function~$f\colon\N\to\N$ such that for every graph~$G\in\mathcal{G}$ it holds that~$\kappa(G)\leq f(\kappa'(G))$ (respectively, we say~$k'$ \emph{upper-bounds}~$k$).
If for two parameters it holds that none of them lower-bounds the other, we say that the two parameters are \emph{unrelated}.

Our work focuses on enumeration, while the great majority of parameterized complexity works study decision (or search and optimization) problems.

\begin{definition}[{\cite[Definition~1]{DBLP:conf/mfcs/CreignouMMSV13}}]
	A \emph{parameterized enumeration problem} is a pair~$(P,\sol)$ such that 
	\begin{itemize}
		\item $P \subseteq \Sigma^* \times \N$ is a parameterized problem over a finite alphabet~$\Sigma$ and
		\item $\sol\colon \Sigma^*\times \N \rightarrow \mathcal{P}(\Sigma^*)$ is a function such that for all~$(x,k) \in \Sigma^* \times \N$, $\sol(x,k)$ is a finite set and~$\sol(x,k) \ne \emptyset \iff (x,k) \in P$. 
	\end{itemize}
\end{definition}
Intuitively, the function~$\sol$ contains for each instance~$(x,k)$ of~$P$ the set of all solutions. Given an instance~$(x,k)$, the task is then to compute~$\sol(x,k)$.

\section{New Notions of Hardness and Kernelization}
\label{sec:concepts}
In this section we introduce two new notions and demonstrate their usefulness. 
The first notion is a many-one reduction that relates a parameterized problem to its unparameterized counterpart.
We call it \emph{``General-Problem-hardness''} as it proves the parameterized version to be as hard as the unparameterized (general) problem.
The second concept is an adaption of an existing kernelization concept for enumeration problems.
It uses some additional space in order to avoid encoding everything in the kernel instance.

\subsection{Computational Hardness} \label{sec:hardness} 
We show hardness for the following problem.
\problemdef{Triangle Detection (\tridet{})}
	{An undirected graph~$G$.}
	{Does~$G$ contain a triangle?}
Since \tridet{} is a special case of \trenum{}, it follows that any lower bound for \tridet{} implies the same lower bound for \trenum{}. Thus, if a certain parameter does not admit a solving algorithm for \tridet{} in some (parameterized) time~$X$, then \trenum{} does not either.

Before giving a formal definition of our concept, consider as an introductory example the parameter minimum degree. 
Adding an isolated vertex to any graph in constant time leaves the set of triangles unchanged and the resulting graph has minimum degree zero.
Hence, one cannot use the parameter minimum degree to design faster algorithms for \tridet{}. 
Upon this trivial example, we study which parameters for \tridet{} cannot be used to design linear-time FPT algorithms under the hypothesis that \tridet{} is not linear-time solvable~\cite{DBLP:conf/focs/AbboudW14}.
To this end we reduce an arbitrary instance of \tridet{} 
\begin{inparaenum}[(i)]
  \item in linear time to
  \item a new equivalent instance of some parameterized version of the problem such that
  \item the parameter is upper-bounded by a constant.
\end{inparaenum}
The corresponding notion of a many-one reduction (also for non-linear running times) is as follows.

\begin{definition} 
	\label{def:k-GP-hard}
	Let~$P \subseteq \Sigma^* \times \N$ be a parameterized problem, let~$Q \subseteq \Sigma^*$ be the unparameterized problem associated with~$P$, and let~$f\colon \N \rightarrow \N$ be a polynomial function.
	We call~$P$ \emph{$\ell$-General-Problem-hard$(f)$ ($\ell$-GP-hard$(f)$)} if there exists an algorithm~$\cal{A}$ transforming any input instance~$x$ of~$Q$ into an instance~$(x',k')$ of~$P$ such that
	\begin{enumerate}[\hspace{3pt}({G}1)]
		\item $\cal{A}$ runs in~$O(f(|x|))$ time,\label{prop:GPh-time}
		\item $x \in Q \iff (x',k') \in P$,\label{prop:GPh-equiv}
		\item $k' \leq \ell$, and\label{prop:GPh-param}
		\item $|x'| \in O(|x|)$.\label{prop:GPh-size}
	\end{enumerate}
	We call~$P$ \emph{General-Problem-hard$(f)$ (GP-hard$(f)$)} if there exists an integer~$k$ such that~$P$ is~$k$-GP-hard$(f)$. 
	We omit the running time and call~$P$ \emph{$k$-General-Problem-hard ($k$-GP-hard)} if~$f$ is a linear function.
\end{definition}
%

Let~$P$ be some parametrized problem being $\ell$-GP-hard($f$) for some polynomial~$f$.
Suppose that we can exclude an algorithm solving~$Q$, the unparameterized version of~$P$, in~$O(f(|x|))$ time under some assumption~$A$.
Since~$P$ is~$\ell$-GP-hard($f$), there is an algorithm computing for any instance~$x$ of~$Q$ an equivalent instance~$(x',k')$ of~$P$ in~$O(f(|x|))$-time.
Then, under assumption~$A$, we can exclude the existence of a~$g(k) \cdot f(|x|)$-time algorithm for~$P$ for any computable function~$g$ due to the following.
Since $k'\leq \ell$, with~$\ell$ being a constant, we have~$g(k)\in O(1)$.
Since~$|x'|\in O(|x|)$ and~$f$ is some polynomial, we have~$f(|x'|)) \in O(f(|x|))$.
Altogether, the two algorithms provide an~$O(f(|x|))$-time algorithm, breaking assumption~$A$.

In a nutshell, a parameterized problem~$P$ being GP-hard is (unconditionally) at least as hard to solve as its unparameterized problem associated with~$P$, formally:


\begin{lemma}
	\label{lem:gph}
	Let~$f\colon \N \rightarrow \N$ be a function, let~$P \subseteq \Sigma^* \times \N$ be a parameterized problem that is~$\ell$-GP-hard$(f)$, and let~$Q \subseteq \Sigma^*$ be the unparameterized problem associated with~$P$.
	If there is an algorithm solving each instance~$(x, k)$ of~$P$ in~$g(k) \cdot f(|x|)$ time, then there is an algorithm solving each instance~$x'$ of~$Q$ in~$O(f(|x'|))$ time. 
\end{lemma}

\begin{proof}
	Assume that there is an algorithm~$\cal{B}$ which solves each instance~$(x, k)$ of~$P$ in ${O(g(k) \cdot f(|x|))}$~time. 
	Let~$x_Q$ be an arbitrary instance of~$Q$. 
	Since~$P$ is~$\ell$-GP-hard$(f)$, there is an algorithm~$\cal{A}$ which transforms~$x$ into a new instance~$(x', k')$ of~$P$ in~$O(f(|x|))$ time such that~$k' \leq \ell$,~$|x'| \in O(|x|)$, and~${x \in Q}$ if and only if~$(x',k') \in P$.
	
	By assumption, algorithm~$\cal{B}$ solves~$(x', k')$ in~\mbox{$g(k') \cdot f(|x'|)$} time.	 
	Since $k' \leq \ell$ and~$|x'| \in O(|x|)$, it holds that~$g(k') \cdot f(|x'|) \in O(f(|x|))$. 
	Since~$x \in Q$ if and only if~$(x',k') \in P$ holds, this algorithm solves~$Q$ in~$O(f(|x|))$ time.
\end{proof}

It is folklore that \tridet{} in tripartite graphs is as hard as the general case (see, e.g., \cite[Section~5]{DBLP:conf/icalp/BjorklundPWZ14}).
We show that \tridet{} with respect to each of the parameters domination number, chromatic number, and diameter is $3$-GP-hard.
Indeed, we also show that \tridet{} is~$9$-GP-hard for the sum of the three parameters.
The \emph{domination number} of a graph is the size of a minimum cardinality set~$S$ with~$\bigcup_{v\in S} N(v) \cup S = V$.
The \emph{chromatic number} of a graph is the minimum number of colors needed to color the vertices such that no edge contains vertices of the same color.
The \emph{diameter} of a graph is the length of the longest shortest path between two vertices.

\begin{proposition}
	\label{chromnum}
	\tridet{} is~$3$-GP-hard with respect to each domination number, chromatic number, and diameter.
	Moreover, \tridet{} is~$9$-GP-hard with respect to the sum of domination number, chromatic number, and diameter.
\end{proposition}

\begin{proof}
	Let~$G=(V,E)$ be an instance of \textsc{Triangle Detection}.
	Let~$V=\{v_1, \ldots, v_n\}$.
	We construct a graph~${G'=(V',E')}$ in time linear in the size of~$G$ such that~$G$ contains a triangle if and only if~$G'$ contains a triangle.
	Moreover,~$G'$ has domination number, chromatic number, and diameter at most three.

	We refer to~\cref{tcolor} for an illustrative example of the following construction of~$G'$.
	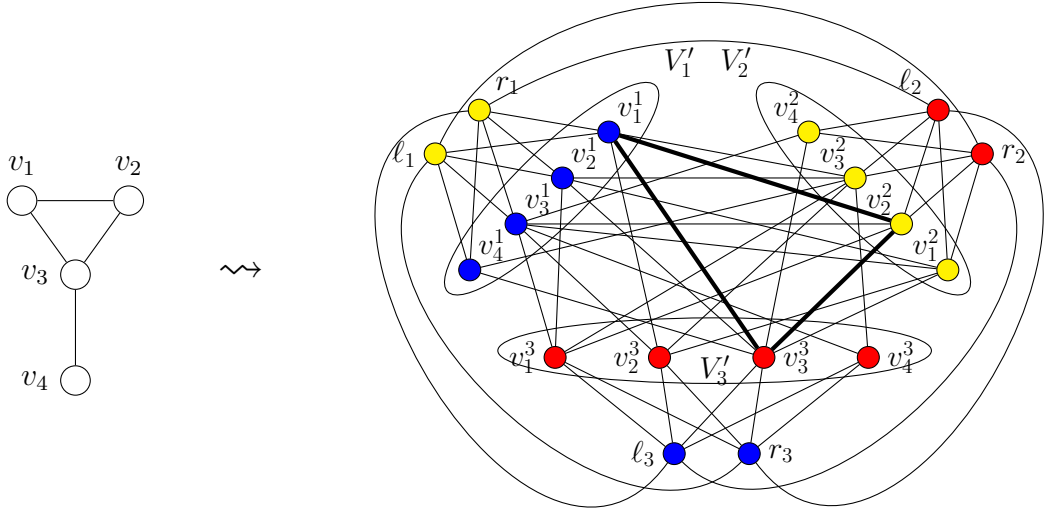
\begin{figure}[t!]
	\centering
	\begin{subfigure}{0.1\textwidth}\centering
		\begin{tikzpicture}[scale=0.8]
		\node[circle, draw, label=$v_1$] (v1) {};
		\node[circle, draw, label=$v_2$] (v2) [right= of v1] {};
		\draw (v1) -- (v2);
		\node[circle, draw, label=left:$v_3$, xshift=20pt, yshift=12pt] (v3) [below= of v1] {}
		edge (v1)
		edge (v2);
		\node[circle, draw, label=left:$v_4$] (v4) [below= of v3] {}
		edge (v3);
		\end{tikzpicture}
	\end{subfigure}
	\hspace*{1cm}
	\Large{$\leadsto$}
	\begin{subfigure}{0.7\textwidth}\centering
		\begin{tikzpicture}[node distance=0.4cm and 0.4cm, every node/.style={scale=0.55}]
		\node[circle, draw, label={[label distance=-7pt]above right:{\LARGE$v_1^1$}}, fill=blue] (V11) {};
		\node[circle, draw, label={[label distance=-7pt]above right:\LARGE{$v_2^1$}}, fill=blue] (V12) [below left= of V11] {};
		\node[circle, draw, label={[label distance=-7pt]above right:\LARGE{$v_3^1$}}, fill=blue] (V13) [below left= of V12] {};
		\node[circle, draw, label={[label distance=-7pt]above right:\LARGE{$v_4^1$}}, fill=blue] (V14) [below left= of V13] {};
		\node[circle, xshift=100pt, draw, label={[label distance=-7pt]above left:\LARGE{$v_4^2$}}, fill=yellow] (V24) [right= of V11] {};
		\node[circle, draw, label={[label distance=-7pt]above left:\LARGE{$v_3^2$}}, fill=yellow] (V23) [below right= of V24]{};
		\node[circle, draw, label={[label distance=-7pt]above left:\LARGE{$v_2^2$}}, fill=yellow] (V22) [below right= of V23]{};
		\node[circle, draw, label={[label distance=-7pt]above left:\LARGE{$v_1^2$}}, fill=yellow] (V21) [below right= of V22] {};
		\node[circle, xshift=-5pt, yshift= -60pt, label=left:\LARGE{$v_1^3$}, draw, fill=red] (V31) [below right= of V13] {};
		\node[circle, xshift= 35pt, label=left:\LARGE{$v_2^3$}, draw, fill=red] (V32) [right= of V31] {};
		\node[circle, xshift= 35pt, label={[label distance=-0pt]0:\LARGE{$v_3^3$}}, draw, fill=red] (V33) [right= of V32] {};
		\node[circle, xshift= 35pt, label=right:\LARGE{$v_4^3$}, draw, fill=red] (V34) [right= of V33] {};
		\node[ellipse, label= right:\LARGE{$V'_1$}, draw, rotate=45, minimum width=200pt, minimum height= 50pt, xshift=-55pt, yshift=-0pt] (V'1) {};
		\node[ellipse, label= right:\LARGE{$V'_2$}, draw, rotate=135, minimum width=200pt, minimum height= 50pt, xshift=-150pt, yshift=-95pt] (V'2) {};
		\node[ellipse, label= {[label distance=-25pt]below:\LARGE{$V'_3$}}, draw, minimum width=295pt, minimum height= 45pt, xshift=72pt, yshift=-150pt] (V'3) {};

		\draw[ultra thick] (V11) -- (V22);
		\draw (V11) -- (V23);
		\draw (V11) -- (V32);
		\draw[ultra thick] (V11) -- (V33);
		\draw (V12) -- (V21);
		\draw (V12) -- (V23);
		\draw (V12) -- (V31);
		\draw (V12) -- (V33);
		\draw (V13) -- (V21);
		\draw (V13) -- (V22);
		\draw (V13) -- (V31);
		\draw (V13) -- (V32);
		\draw (V21) -- (V32);
		\draw (V21) -- (V33);
		\draw (V22) -- (V31);
		\draw[ultra thick] (V22) -- (V33);
		\draw (V23) -- (V31);
		\draw (V23) -- (V32);
		\draw (V13) -- (V24);
		\draw (V13) -- (V34);
		\draw (V14) -- (V23);
		\draw (V14) -- (V33);
		\draw (V23) -- (V34);
		\draw (V24) -- (V33);

		\node[circle, draw, label=left:\LARGE{$\ell_1$}, xshift= -55pt, yshift=-15pt, fill=yellow] (l1) [above left= of V12] {};
		\node[circle, draw, label= 135:\LARGE{$\ell_2$}, xshift= 25pt, yshift=15pt, fill=red] (l2) [above right= of V23] {};
		\node[circle, draw, label=left:\LARGE{$\ell_3$}, yshift=-30pt, xshift=10pt, fill=blue] (l3) [below= of V32] {};

		\node[circle, draw, label=45:\LARGE{$r_1$}, xshift= -25pt, yshift=15pt, fill=yellow] (r1) [above left= of V12] {};
		\node[circle, draw, label=right:\LARGE{$r_2$}, xshift= 55pt, yshift=-15pt, fill=red] (r2) [above right= of V23] {};
		\node[circle, draw, label=right:\LARGE{$r_3$}, yshift=-30pt, xshift=-10pt, fill=blue] (r3) [below= of V33] {};
		\draw (l1) -- (V11);
		\draw (l1) -- (V12);
		\draw (l1) -- (V13);
		\draw (l1) -- (V14);
		\draw (l2) -- (V21);
		\draw (l2) -- (V22);
		\draw (l2) -- (V23);
		\draw (l2) -- (V24);
		\draw (l3) -- (V31);
		\draw (l3) -- (V32);
		\draw (l3) -- (V33);
		\draw (l3) -- (V34);

		\draw (r1) -- (V11);
		\draw (r1) -- (V12);
		\draw (r1) -- (V13);
		\draw (r1) -- (V14);
		\draw (r2) -- (V21);
		\draw (r2) -- (V22);
		\draw (r2) -- (V23);
		\draw (r2) -- (V24);
		\draw (r3) -- (V31);
		\draw (r3) -- (V32);
		\draw (r3) -- (V33);
		\draw (r3) -- (V34);

		\draw[bend right=30] (l2) edge (r1);
		\draw[bend right=90] (l1) edge (r3);
		\draw[bend left=65] (l1) edge (r2);
		\draw[bend right=90] (l3) edge (r2);
		\begin{pgfinterruptboundingbox}
		\draw[bend left=115,looseness=1.55] (l2) edge (r3);
		\draw[bend left=-115,looseness=1.55] (r1) edge (l3);
		\end{pgfinterruptboundingbox}
		\end{tikzpicture}
	\end{subfigure}
	\caption{An illustration of the construction of the graph in the proof of~\cref{chromnum}, exemplified on the left-hand side graph. 
			The triangle~$\{v_1, v_2, v_3\}$ on the left-hand side corresponds in the right-hand side to multiple triangles, one of them is highlighted by bold lines.}
	\label{tcolor}
	\end{figure}
	Let~$G'$ be initially empty.
	Add three copies~$V'_1,V'_2,V'_3$ of~$V$ to~$G'$, and let~$V'_i=\{v_1^i, \ldots, v_n^i\}$, $i~\in~\{1,2,3\}$.
	For each edge~$\{v_x,v_y\}\in E$, add the edge set~${\{\{v_x^i,v_y^j\}\mid i,j\in\{1,2,3\},i\neq j\}}$ to~$G'$.
	Add the vertex sets~\mbox{$L = \{\ell_1, \ell_2, \ell_3\}$} and~${R=\{r_1, r_2, r_3\}}$ to~$G'$.
	For each~${i\in\{1,2,3\}}$, connect~$\ell_i$,~$r_i$ with each vertex in~$V_i'$ by an edge, that is, add~$\{\{\ell_i, v_x^i\},\{r_i, v_x^i\} \mid i \in \{1,2,3\}, x \in [n]\}$ to the edge set of~$G'$.
	Finally, for ${i\neq j}$, connect~${\ell_i\in L}$ with~${r_j\in R}$ by an edge in $G'$, that is, add the edge set~$\{\{l_i, r_j\} \mid i,j\in\{1,2,3\},i \neq j\}$.
	This completes the construction of~$G'$.

	We prove that the properties \condRef{G}{prop:GPh-time}--\condRef{G}{prop:GPh-size} of \cref{def:k-GP-hard} are fulfilled.
	Note that~$G'$ is constructed in~$O(|G|)$~time~\condRef{G}{prop:GPh-time} and contains~$3n+6$ vertices and ${6m+6n+6}$~edges~\condRef{G}{prop:GPh-size}.
	Observe that each~\mbox{$V_i'$,~$i\in\{1,2,3\}$}, forms an independent set in~$G'$.
	In addition, each vertex set~$N(\ell_i)$,~$N(r_i)$,~$i\in\{1,2,3\}$, forms an independent set in~$G'$.

	Next, we prove that~$G$ contains a triangle if and only if~$G'$ contains a triangle~\condRef{G}{prop:GPh-equiv}.
	Suppose that~$\{v_x, v_y, v_z\}$ forms a triangle in~$G$. 
	Then, by construction of~$G'$,~$G'$~contains the vertices~$v_x^1$,~$v_y^2$, and~$v_z^3$ and the edges~$\{v_x^1, v_y^2\}$,~$\{v_y^2, v_z^3\}$, and~$\{v_x^1, v_z^3\}$. 
	Hence,~$\{v_x^1, v_y^2,v_z^3\}$ forms a triangle in~$G'$.

	Conversely, suppose that~$\{x,y,z\}$ forms a triangle in~$G'$.
	As~$V_i'$, $N(\ell_i)$, and~$N(r_i)$ form independent sets for each~$i\in\{1,2,3\}$ in~$G'$,~$\{x,y,z\}$ does not contain~$\ell_i$,~$r_i$, or at least two vertices of~$V_i'$,~$i\in\{1,2,3\}$.
	It follows that for each~$i\in\{1,2,3\}$ it holds~\mbox{$|\{x,y,z\}\cap V_i'|=1$}.
	Let without loss of generality be~$x = v_a^1 \in V'_1$,~$y =v_b^2 \in V'_2$, and~$z = v_c^3 \in V'_3$. 
	As for~$i \neq j$ it holds that~$\{v_a^i,v_b^j\} \in E'$ if and only if~$\{v_a, v_b\}\in E$, it follows that~$\{v_a, v_b, v_c\}$ forms a triangle in~$G$.

	Last, we prove that the domination number, the chromatic number, and the diameter of~$G'$ are all at most~$3$~\condRef{G}{prop:GPh-param}.
	As for each~$i\in\{1,2,3\}$, each vertex in~$V'_i$~is connected with~$\ell_i$, and~$N(r_i)\cap L\neq \emptyset$, the set~$L$ is a dominating set in~$G'$.
	As each~$V'_i$,~$i\in\{1,2,3\}$, is an independent set in~$G'$, color the vertices of~$V_i'$ with color~$i$. 
	Next, color the vertices~$\ell_i, r_i$ with color~$1+(i\mod 3)$ for each~$i\in\{1,2,3\}$. 
	This forms a valid coloring of the vertices in~$G'$ with at most three colors.
	Observe that since~$V_i\cup R\setminus\{r_i\}\in N(\ell_i)$ and~$V_i\cup L\setminus\{\ell_i\}\in N(r_i)$, each vertex in~$G'$ has distance at most two to~$\ell_i$ and~$r_i$, for each~$i\in\{1,2,3\}$.
	As each vertex in~$V_1'\cup V_2'\cup V_3'$ has a neighbor in~$L\cup R$, it follows that~$G'$ is of diameter at most three.

	Altogether, \condRef{G}{prop:GPh-time}--\condRef{G}{prop:GPh-size} of \cref{def:k-GP-hard} are satisfied and hence the proposition follows.
\end{proof}

\subsection{Enum-Advice Kernelization}
The second new notion we introduce in this paper is an adaption of an enumerative kernelization concept due to Creignou et al.~\cite{DBLP:conf/mfcs/CreignouMMSV13}.

The aim of kernelization is to efficiently reduce a large instance of a computationally hard, say~$\NP$-hard, problem to an ``equivalent'' instance (called ``kernel'') whose size only depends on the parameter and not on the size of the original input instance.
Then, solving the kernel by a trivial brute-force algorithm often significantly reduces the overall running time. 
This technique is by no means restricted to computationally hard problems even though it was invented to tackle problems for which no polynomial-time algorithms are known.

Observe that kernelization is usually defined for decision problems only.
Creignou et al.~\cite{DBLP:conf/mfcs/CreignouMMSV13} developed a concept to address enumeration problems.
Roughly speaking, their concept requires that \emph{all} solutions of the input instance can be recovered from the input instance and the solutions of the kernel 
(see \cref{fig:eak}\subref{sfig:ek}).
We modify the concept by adding a generic object which we call the \emph{advice} of the kernelization.
The intention behind this change is that in order to compute all solutions of the \emph{input} instance, one only needs the kernel and the advice (which might be much smaller than the input instance), see \cref{fig:eak} for an illustration.
In the examples we provide in this paper, in the advice we store information about all triangles that are destroyed by data reduction rules.

\begin{figure}[t]
  \centering
  \begin{subfigure}[b]{0.45\textwidth} \centering
  \begin{tikzpicture}[scale=0.69, every node/.style={scale=0.69}]
  	\def\x{1}
	  \node[para, minimum size = 3cm] (in)  at (0,0) {input~$(x,\ell)$};
	  \node[para, minimum size = 2cm, align=center] (k) at (5-0.5,0) {kernel\\$I(x,\ell)$};
	  \node[para] (s) at (5-0.5,-4+\x) {$W =  \sol(I(x,\ell))$};
	  \node[para] (S) at (0,-4+\x) {$\bigcup\limits_{w \in W}f(x,w) = \sol(x,\ell)$};
	  \draw[->, thick] (k) edge (s);
	  \node [above] at(2.2,0) {$\mathcal{R}$};
	  \draw[->, thick] (in) -- (k);
	  \draw[->, thick, bend left = 40] (1.46,-1.47) to ($(S.north east) + (-0.05,-0.05)$);
	  \draw[->, thick, bend right = 40] ($(s.north west) + (0.05,-0.05)$) to ($(S.north east) + (-0.05,-0.05)$) node [above right= 0.15cm and 0cm] {$\mathcal{T}_f$};
	  \draw[decorate, decoration={brace, amplitude=0.3cm}] (4-0.5, 0.91) -- (6-0.5, 0.91) node[midway, above=0.3cm] {$\leq h(\ell)$};  
  \end{tikzpicture}
  \subcaption{Enum kernelization.}\label{sfig:ek}\end{subfigure}
\hfill
  \begin{subfigure}[b]{0.5\textwidth}\centering
  \begin{tikzpicture}[scale=0.69, every node/.style={scale=0.69}]
  	\def\x{1}
	  \node[para, minimum size = 3cm] (in)  at (0,0) {input~$(x,\ell)$};
	  \node[para, minimum size = 2cm, align=center] (k) at (6,0) {kernel\\~$I(x,\ell)$};
	  \node[para, minimum size = 1.5cm, align=center, dashed] (A) at (4.25,-0.25) {advice\\~$A(x,\ell)$};
	  \node[para] (s) at (6,-4+\x) {$W = \sol(I(x,\ell))$};
	  \node[para] (S) at (1,-4+\x) {$\bigcup\limits_{w \in W}f(w, A(x,\ell)) = \sol(x,\ell)$};
	  \draw[->, thick] (k) to (s);
	  \node [below] at(2.7,0.75) {$\mathcal{R}$};
	  \draw[->, thick] (1.5,0.75) -- (5,0.75);
	  \draw[->, thick, bend left=5] (1.8,0.75) to (4.25,0.5);
	  \draw[->, thick, bend left = 15] (A) to ($(S.north east) + (-0.05,-0.05)$);
	  \draw[->, thick, bend right = 60] ($(s.north west) + (0.05,-0.05)$) to ($(S.north east) + (-0.05,-0.05)$) node [above right= 0.2cm and 0.4cm] {$\mathcal{T}_f$};
	  \draw[decorate, decoration={brace, amplitude=0.3cm}] (5, 0.91) -- (7, 0.91) node[midway, above=0.3cm] {$\leq h(\ell)$};
  \end{tikzpicture}
  \subcaption{Enum-advice kernelization.}\label{sfig:eak}\end{subfigure}
	\caption{
		A schematic picture of \subref{sfig:ek} enum- 
		and \subref{sfig:eak} enum-advice 
		kernelization. 
		Here, $\mathcal{R}$~refers to the algorithm that produces the kernel and, for enum-advice kernelization, also the advice, and~$\mathcal{T}_f$ refers to the polynomial-delay algorithm enumerating all solutions of the input.
	}
\label{fig:eak}
\end{figure}
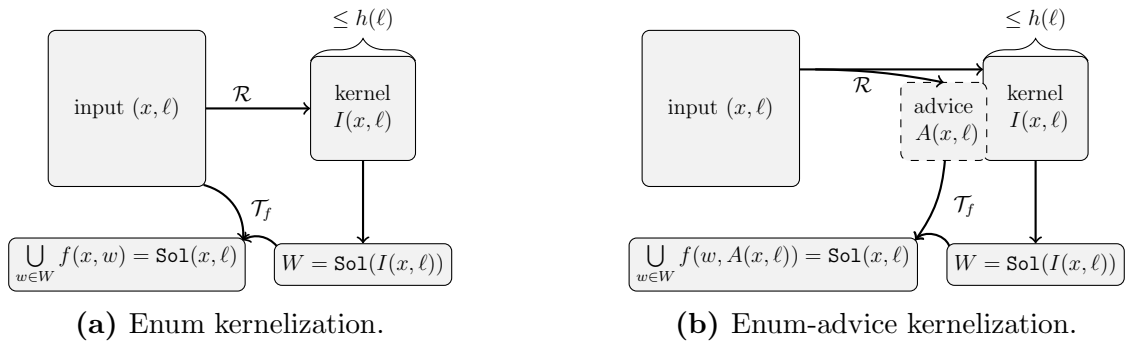
We will now give a formal definition of our new enumerative kernelization concept and then discuss the potential advantages compared to the concept by Creignou~et~al.~\cite{DBLP:conf/mfcs/CreignouMMSV13}.

\begin{definition}
	\label{def:eakernel}
	Let~$(P, \sol)$ be a parameterized enumeration problem.
	Let~$\mathcal{R}$ be an algorithm which for every input~$(x, k)$ computes in time polynomial in~$|x|+k$ a pair~$(I(x,k), \allowbreak A(x,k))$. 
	We call~$\mathcal{R}$ an \emph{enum-advice kernelization} of~$(P, \sol)$ if
	\begin{enumerate}[\hspace{3pt}({K}1)]
		\item there exists a function~$h$ such that for all~$(x,k)$ it holds that~$|I(x,k)| \leq h(k)$,\label{prop:eak-small}
		\item for all~$(x,k)$ it holds that~$(x,k) \in P \iff I(x,k) \in P$, and \label{prop:eak-correct} 
		\item there exists a function~$f$ such that for all~$(x,k) \in P$ \label{prop:eak-enum} 
		\begin{enumerate}
			\item~$\forall p,q \in \sol(I(x,k)) \colon p \neq q \implies f(p, A(x,k)) \cap f(q, A(x,k)) = \emptyset$, \label{prop:eak-enum-disjoint}
			\item~$\bigcup_{w \in \sol(I(x,k))} f(w, A(x,k)) = \sol(x,k)$, and \label{prop:eak-enum-all}
			\item there exists an algorithm~$\mathcal{T}_f$ such that for every~$(x,k) \in P$ and ${w \in \sol(I(x,k))}$, ${\mathcal{T}_f}$ computes~$f(w,A(x,k))$ in FPT-delay time~\cite{DBLP:conf/mfcs/CreignouMMSV13}.\label{prop:eak-enum-algo}
		\end{enumerate}
		\end{enumerate}
	If~$\mathcal{R}$ is an enum-advice kernelization of~$(P, \sol)$, then~$I(x,k)$ is called the \emph{kernel} of~$(x,k)$ and~$A(x,k)$ is called the \emph{advice} of~$I(x,k)$.
	If algorithm~$\mathcal{T}_f$ has~$p$-delay time for some polynomial~$p$ $($only in~$|x|)$, then we say that the problem admits a \emph{$p$-delay enum-advice kernel}.
\end{definition}
Clearly, since every polynomial-time solvable enumeration problem has a trivial enum-advice kernelization, we are only interested in those kernelizations where~$\mathcal{R}$ and~$\mathcal{T}_f$ are both significantly faster than the best (known) algorithms to solve the enumeration problem.

We will now discuss the potential advantages of our new definition compared to enum-kern\-eli\-za\-tion. 
The advice can be used to design faster algorithms since the advice might be much smaller than the input instance, as described in the following example.
Observe that one can set~$A(x, \ell) = (x,\ell)$, and thus enum-advice kernelization is a generalization of enum-kernelization.
  \begin{example}
	  Consider the \textsc{Enum Vertex Cover} problem parameterized by solution size~$k$; here the task is to list all minimal vertex covers of size at most~$k$ of an input graph~$G$. 
	  As observed by Creignou et al. \cite[Proposition~1]{DBLP:conf/mfcs/CreignouMMSV13}, the standard Buss’ kernelization~\cite{DF13}  provides an enum-kernelization for this problem. 
	  It consists of the following two data reduction rules:
	  \begin{enumerate}
		  \item If~$\operatorname{deg}(v) > k$, then remove~$v$ and all edges incident with~$v$ from~$G$ and decrease~$k$ by one ($v$ is contained in every vertex cover of size at most~$k$).
		  \item If~$\operatorname{deg}(v) = 0$, then remove~$v$ from~$G$ (no minimal vertex cover contains~$v$).
	  \end{enumerate}
	  Let~$G'$ be the graph obtained from exhaustively applying the two data reduction rules on~$G$ such that none of the two rules is applicable to~$G'$.
	  If~$G'$ contains more than~$k^2$ edges, then return the complete graph of~$k+2$ vertices as there is no vertex cover of size~$k$.
	  Otherwise, return graph~$G'$.
	  
	  Let~$V_D$ be the set of vertices that are deleted by the first rule.
	  The set of minimal vertex covers in the input graph~$G$ can be obtained by adding~$V_D$ to each minimal vertex cover in the kernel~$G'$.
	  For the set of minimal vertex covers of~$G$ of size at most~$k$ one considers only minimal vertex covers in~$G'$ that have size at most~$k-|V_D|$.
	  In an enum-advice kernel one can store~$V_D$ in the advice, which then has size~$O(k)$.
	  To compute all minimal vertex covers in~$G$, we compute all minimal vertex covers in~$G'$ and add~$V_D$ to each.
	  As~$|V_D|\leq k$, we can add~$V_D$ to a computed minimal vertex cover of~$G'$ in~$O(k)$ time.
	  In contrast, the enum-kernelization concept would require to (re-)compute~$V_D$ for each minimal vertex cover in~$G'$ which requires time linear in the size of~$G$.
	  We refer to \cref{fig:enumvc} for an illustration.
	  When considering~$\NP$-hard problems replacing a term like~$n+m$ by~$k$ might be only a small improvement, but for polynomial-time solvable problems (like \trenum{}) it can have a significant impact.
	  
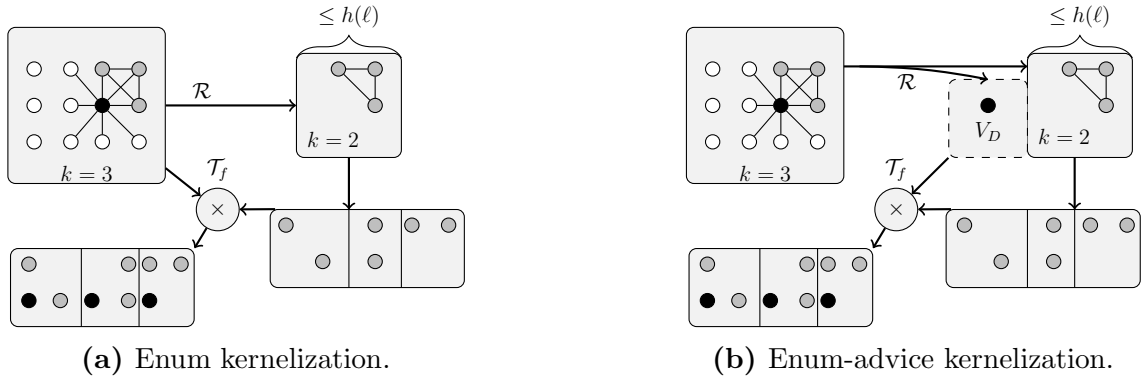
\begin{figure}[t]
  \centering
  \begin{subfigure}[b]{0.45\textwidth}\centering
  \begin{tikzpicture}[scale=0.69, every node/.style={scale=0.69}]
  	\def\x{1}  	
  	\def\ysh{0.0}
  	\def\yshx{0.75}
	  \node[para, minimum size = 3cm] (in)  at (0,0) {};
	  {
		\node[wnode] (v1) at (-1,0) {};
		\node[wnode] (v2) at (-1,0.7) {};
		\node[wnode] (v3) at (-1,-0.7) {};
		\node[wnode] (v4) at (-0.3,0) {};
		\node[wnode] (v5) at (-0.3,0.7) {};
		\node[wnode] (v6) at (-0.3,-0.7) {};
		\node[bnode] (v7) at (0.3,0) {};
		\node[gnode] (v8) at (0.3,0.7) {};
		\node[wnode] (v9) at (0.3,-0.7) {};
		\node[gnode] (v10) at (1,0) {};
		\node[gnode] (v11) at (1,0.7) {};
		\node[wnode] (v12) at (1,-0.7) {};
		\node (vk) at (0,-1.3) {$k=3$};

		\draw (v7) -- (v4);
		\draw (v7) -- (v5);
		\draw (v7) -- (v6);
		\draw (v7) -- (v8);
		\draw (v7) -- (v9);
		\draw (v7) -- (v10);
		\draw (v7) -- (v11);
		\draw (v7) -- (v12);
		
		\draw (v8) -- (v11);
		\draw (v10) -- (v11);
		\draw (v10) -- (v8);
	  }
	  \node[para, minimum size = 2cm, align=center] (k) at (5,0) {};
	  {
		\node[gnode] (v8) at (4.8,0.7) {};
		\node[gnode] (v10) at (5.5,0) {};
		\node[gnode] (v11) at (5.5,0.7) {};
		\node (vk) at (4.7,-0.65) {$k=2$};
	
		\draw (v8) -- (v11);
		\draw (v10) -- (v11);
		\draw (v10) -- (v8);
	  }

	    \node[para, minimum height = 1.5cm, minimum width = 3.7cm] (s) at (5.35,-2.75-\ysh) {};

	  {
		\node[gnode] (v11) at (6.2,-2.3-\ysh) {};
		\node[gnode] (v12) at (6.9,-2.3-\ysh) {};
		\draw (6,-2-\ysh) -- (6,-3.5-\ysh);
	  }
	  {
		\node[gnode] (v10) at (5.5,-3-\ysh) {};
		\node[gnode] (v11) at (5.5,-2.3-\ysh) {};
		\draw (5,-2-\ysh) -- (5,-3.5-\ysh);
	  }
	  {
		\node[gnode] (v8) at (3.8,-2.3-\ysh) {};
		\node[gnode] (v10) at (4.5,-3-\ysh) {};
	  }

	  \node[para, minimum height = 1.5cm, minimum width = 3.5cm] (S) at (0.3,-2.75-\yshx) {};

	  {
	  	\node[bnode] (v8) at (1.2,-3-\yshx) {};
		\node[gnode] (v11) at (1.8,-2.3-\yshx) {};
		\node[gnode] (v12) at (1.2,-2.3-\yshx) {};
		\draw (1,-2-\yshx) -- (1,-3.5-\yshx);
	  }
	  {
	  	\node[bnode] (v8) at (0.1,-3-\yshx) {};
		\node[gnode] (v10) at (0.8,-3-\yshx) {};
		\node[gnode] (v11) at (0.8,-2.3-\yshx) {};
		\draw (-0.1,-2-\yshx) -- (-0.1,-3.5-\yshx);
	  }
	  {
	  	\node[bnode] (v8) at (-1.1,-3-\yshx) {};
		\node[gnode] (v8) at (-1.1,-2.3-\yshx) {};
		\node[gnode] (v10) at (-0.5,-3-\yshx) {};
	  }

	  \draw[->, thick] (k) edge (5,-2);
	  \node [above] at(2.2,0) {$\mathcal{R}$};
	  \draw[->, thick] (in) -- (k);
	  \node (Tf) at (2.5,-2)[para,label=90:{$\mathcal{T}_f$},draw,circle]{$\times$};
	  \draw[->, thick] ($(s.north west)-(-0.1,0)$) to (Tf);
	  \draw[->, thick] (in) to (Tf);
	  \draw[->, thick] (Tf) to (S.north east);
	  \draw[decorate, decoration={brace, amplitude=0.3cm}] (4, 0.91) -- (6, 0.91) node[midway, above=0.3cm] {$\leq h(\ell)$};  
  \end{tikzpicture}
  \subcaption{Enum kernelization.}\label{sfig:ekvc}\end{subfigure}
  \hfill
  \begin{subfigure}[b]{0.45\textwidth}\centering
  \begin{tikzpicture}[scale=0.69, every node/.style={scale=0.69}]
  	\def\x{1}
  	\def\ysh{0.0}
  	\def\yshx{0.75}
	  \node[para, minimum size = 3cm] (in)  at (0,0) {};
	  {
		\node[wnode] (v1) at (-1,0) {};
		\node[wnode] (v2) at (-1,0.7) {};
		\node[wnode] (v3) at (-1,-0.7) {};
		\node[wnode] (v4) at (-0.3,0) {};
		\node[wnode] (v5) at (-0.3,0.7) {};
		\node[wnode] (v6) at (-0.3,-0.7) {};
		\node[bnode] (v7) at (0.3,0) {};
		\node[gnode] (v8) at (0.3,0.7) {};
		\node[wnode] (v9) at (0.3,-0.7) {};
		\node[gnode] (v10) at (1,0) {};
		\node[gnode] (v11) at (1,0.7) {};
		\node[wnode] (v12) at (1,-0.7) {};
		\node (vk) at (0,-1.3) {$k=3$};

		\draw (v7) -- (v4);
		\draw (v7) -- (v5);
		\draw (v7) -- (v6);
		\draw (v7) -- (v8);
		\draw (v7) -- (v9);
		\draw (v7) -- (v10);
		\draw (v7) -- (v11);
		\draw (v7) -- (v12);
		
		\draw (v8) -- (v11);
		\draw (v10) -- (v11);
		\draw (v10) -- (v8);
	  }
	  \node[para, minimum size = 2cm, align=center] (k) at (6,0) {};
	  
	  {
		\node[gnode] (v8) at (5.8,0.7) {};
		\node[gnode] (v10) at (6.5,0) {};
		\node[gnode] (v11) at (6.5,0.7) {};
		\node (vk) at (5.7,-0.6) {$k=2$};
	
		\draw (v8) -- (v11);
		\draw (v10) -- (v11);
		\draw (v8) -- (v10);
	  }
	  
	  \node[para, minimum size = 1.5cm, align=center, dashed] (A) at (4.25,-0.25) {};
	  
	  {
		\node[bnode] (v8) at (4.25,0) {};
		\node (vD) at (4.25,-0.5) {$V_D$};
	  }
	  
	  \node[para, minimum height = 1.5cm, minimum width = 3.7cm] (s) at (5.3,-2.75-\ysh) {};

	  {
		\node[gnode] (v11) at (6.2,-2.3-\ysh) {};
		\node[gnode] (v12) at (6.9,-2.3-\ysh) {};
		\draw (5.9,-2-\ysh) -- (5.9,-3.5-\ysh);
	  }
	  {
		\node[gnode] (v10) at (5.5,-3-\ysh) {};
		\node[gnode] (v11) at (5.5,-2.3-\ysh) {};
		\draw (5,-2-\ysh) -- (5,-3.5-\ysh);
	  }
	  {
		\node[gnode] (v8) at (3.8,-2.3-\ysh) {};
		\node[gnode] (v10) at (4.5,-3-\ysh) {};
	  }

	  \node[para,minimum height = 1.5cm, minimum width=3.5cm] (S) at (0.3,-2.75-\yshx) {};

	  {
	  	\node[bnode] (v8) at (1.2,-3-\yshx) {};
		\node[gnode] (v11) at (1.8,-2.3-\yshx) {};
		\node[gnode] (v12) at (1.2,-2.3-\yshx) {};
		\draw (1,-2-\yshx) -- (1,-3.5-\yshx);
	  }
	  {
	  	\node[bnode] (v8) at (0.1,-3-\yshx) {};
		\node[gnode] (v10) at (0.8,-3-\yshx) {};
		\node[gnode] (v11) at (0.8,-2.3-\yshx) {};
		\draw (-0.1,-2-\yshx) -- (-0.1,-3.5-\yshx);
	  }
	  {
	  	\node[bnode] (v8) at (-1.1,-3-\yshx) {};
		\node[gnode] (v9) at (-1.1,-2.3-\yshx) {};
		\node[gnode] (v10) at (-0.5,-3-\yshx) {};
	  }

	  \draw[->, thick] (5.9,-1) to (5.9,-2);
	  \node [below] at(2.7,0.75) {$\mathcal{R}$};
	  \draw[->, thick] (1.5,0.75) -- (5,0.75);
	  \draw[->, thick, bend left=5] (1.8,0.75) to (4.25,0.5);
	  \node (Tf) at (2.5,-2)[para,label=90:{$\mathcal{T}_f$},draw,circle]{$\times$};
	  \draw[->, thick] ($(s.north west)-(-0.1,0)$) to (Tf);
	  \draw[->, thick] (A) to (Tf);
	  \draw[->, thick] (Tf) to (S.north east);
	  \draw[decorate, decoration={brace, amplitude=0.3cm}] (5, 0.91) -- (7, 0.91) node[midway, above=0.3cm] {$\leq h(\ell)$};
  \end{tikzpicture}
  \subcaption{Enum-advice kernelization.}\label{sfig:eakvc}\end{subfigure}
	\caption{
		A schematic picture of \subref{sfig:ekvc} enum- and \subref{sfig:eakvc} enum-advice kernelization for \textsc{Enum Vertex Cover}. In this example the kernel only consists of a single vertex while the input instance consists of 12 vertices and 11 edges.
	}
	\label{fig:enumvc}
\end{figure}	  
	  
  \end{example}

As discussed in the above example, the advice can be used to efficiently compute the solutions from the kernel.
In general, enumeration algorithms can be derived from enum-advice kernels as stated in the next lemma.

\begin{lemma}
	\label{lem:kernel-alg}
	Let~$\mathcal{R}$ be an enum-advice kernelization of a parameterized enumeration problem~$(P, \sol)$ such that for every instance~$(x,k)$ of~$P$:
	\begin{itemize}
		\item $\mathcal{R}$ runs in~$O((|x| + k)^c)$ time for some constant~$c$; 
		\item the unparameterized version of~$P$ can be solved in~$g(|x|)$ time on~$x$;
		\item the kernelization computes the pair~$(I,A)$ where~$|I| \leq h(k)$, and algorithm~$\mathcal{T}_f$ takes $O(|I|^d)$ time between generating two solutions for some constant~$d$; and
		\item $\hashs$ denotes the number of solutions in~$I$ and~$\hashS$ denotes the number of solutions in~$x$. 
	\end{itemize}
	Then,~$(P, \sol)$ can be solved in~$O((|x|+k)^c + g(h(k)) + (\hashs + \hashS) \cdot h(k)^d)$ time.
\end{lemma}

\begin{proof}
	We use the notation as stated in the lemma and assume that all stated conditions hold.
	We will prove that there exists an algorithm solving~$(P, \sol)$ in~$O((|x|+k)^c + g(h(k)) + (\hashs + \hashS) \cdot h(k)^d)$ time.

	First, compute the kernel~$(I,A)$ in~$O((|x| + k)^c)$ time. 
	Second, find all solutions in~$I$ in~$g(|I|) \in O(f(h(k)))$ time. 
	Third, compute~$\bigcup_{w \in \sol(I)} f(w, A)$. 
	This can be done by running~$\mathcal{T}_f$ on every solution in~$I$. 
	There are~$\hashs$ solutions in~$I$, hence there are~$\hashs$ iterations of~$\mathcal{T}_f$. 
	The algorithm therefore spends at most~$O(\hashs \cdot |I|^d)$ time as precalculation or postcalculation time.
	We refer to Creignou et al.~\cite{DBLP:conf/mfcs/CreignouMMSV13} for a formal definition of pre- and postcalculation time.
	Informally, precalculation time is the time from the start of the calculation until the first solution is listed and postcalculation time is the time needed after the last solution is listed.
	Apart from pre- and postclaculation time, the time between two solution outputs is at most~$O(|I|^d) \subseteq O(h(k)^d)$. 
	Since there are~$\hashS$ solutions in~$x$, computing all solutions takes~$(\hashs+\hashS) \cdot (h(k))^d$ time. 
	Thus, this algorithm takes~$O((|x|+k)^c + g(h(k)) + (\hashs + \hashS) \cdot h(k)^d)$ time.
\end{proof}

Note that in general we cannot give any meaningful upper bound on the delay of the constructed algorithm as the kernel instance might be packed with solutions~$p$ such that~$f(w,A) = \emptyset$.
If no such solutions exist, then the delay of the described algorithm is~$O((|x|+k)^c + f(h(k)) +  h(k)^d))$.
The delay of all algorithms presented in our work is only upper-bounded by the respective running times of the algorithms.
\label{sec:eakernels}

\section{Algorithms}
\label{sec:algs}
\begin{figure}[t!]
	\centering
	\begin{tikzpicture}[scale=0.85, transform shape]

	\def\x{1.55}
	\def\y{1.55}
	\def\sc{1}
	\def\scx{1.05}
	\tikzstyle{xnode}=[fill=lightgray!30,rounded corners, draw, align=center,scale=\scx]
	\tikzstyle{xxnode}=[fill=white,rounded corners, draw, align=center,scale=\sc]

	\draw[dashed,very thick,color=green!70!black,-] (-7.75*\x,-1.0*\y) rectangle (4.25*\x,4.15*\y);
	\draw[dashed,very thick,color=orange!85!black,-] (-7.75*\x,-1.1*\y) rectangle (4.25*\x,-2.4*\y);
	\draw[dashed,very thick,color=red,-] (-7.75*\x,-2.5*\y) rectangle (4.25*\x,-4.75*\y);
	\draw[dashed,very thick,color=blue,-] (-7.2*\x,0.5*\y) rectangle (3.7*\x,3.3*\y);
	
	\node (baderx) at (0*\x,3.75*\y)[xnode]{Distance to $d$-degenerate + maximum degree~(\Cref{Bader})};

	\node (fes) at (-5.25*\x,2.7*\y)[xnode]{Feedback edge number \\  (\Cref{fesnkernel}, \Cref{cor:fesnalg})};

	\node (dddr) at  (0*\x,1.9*\y)[xnode,minimum height=2.5*\y cm, minimum width=6.5*\x cm]{};
	\node (ddd) at (-0.0*\x,2.8*\y)[rotate=0, scale=1.0]{Distance to $d$-degenerate (\Cref{thm:dtddkernel}, \Cref{cor:dtddalg})};
	\node (vc) at (-0.85*\x,2.3*\y)[xxnode]{Vertex cover number (0-degenerate)};
	\node (fvs) at (-0.25*\x,1.55*\y)[xxnode]{Feedback vertex number (1-degenerate)};
	\node (dots) at  (1*\x,1.1*\y)[scale=1.2]{$\vdots$};

	\node (cog) at (-3.8*\x,-0.25*\y)[xnode]{Distance to \\ Cograph \\ (\Cref{thm:distbipchorcog}\ref{thm:distcog})};

	\node (deg) at (3*\x,-0.5*\y)[xnode]{Degeneracy~\cite{DBLP:journals/siamcomp/ChibaN85}};
	\node (bip) at (0.75*\x,-0.25*\y)[xnode]{Distance to \\ Bipartite \\(\Cref{thm:distbipchorcog}\ref{thm:distbip})};
		\node (chordal) at (-1.3*\x,-0.25*\y)[xnode]{Distance to \\ Chordal \\ (\Cref{thm:distbipchorcog}\ref{thm:distchor})};
		
	\node (cw) at (-2.5*\x,-1.75*\y)[xnode]{Clique-width \\ (\Cref{thm:cliquewidth})};
	\node (avg) at (3*\x,-1.75*\y)[xnode]{Average \\ degree~\cite{DBLP:journals/siamcomp/ItaiR78}};

	\node (dom) at (-5.5*\x,-3.1*\y)[xnode]{Domination \\ Number~(\Cref{chromnum})};
	\node (diam) at (-3.8*\x,-4.1*\y)[xnode]{Diameter~(\Cref{chromnum})};
	\node (chrom) at (-0.0*\x,-3.1*\y)[xnode]{Chromatic \\ Number~(\Cref{chromnum})};
	\node (mindeg) at (1*\x,-4.15*\y)[xnode]{Minimum degree \\ (\Cref{sec:hardness})};

	\draw[thick] (fvs) -- (vc);
	\draw[thick] (fvs.north west) -- (fes);
	\draw[thick] (cw) -- ($(fvs.south west)+(0.45,0)$);
	\draw[thick] (dddr) -- (baderx);
	\draw[thick] (dddr) -- (deg);

	\draw[thick] (deg) -- (avg);
	\draw[thick] (mindeg) -- (avg);

	\draw[thick] (deg) -- (chrom);
	\draw[thick] (deg) -- (avg);
	\draw[thick] (mindeg) -- (avg);

	\draw[thick] (cw) -- (cog);
	\draw[thick] (cog) -- (vc.south west);
	\draw[thick] (diam) -- (cog);
	\draw[thick] (diam) -- (dom);
	
	\draw[thick] (chrom) -- (mindeg);
	\draw[thick] (bip) -- (fvs);
	\draw[thick] (bip) -- (chrom);
	\draw[thick] (fvs) -- (chordal);
	
	\node at (-6.3*\x,-0.65*\y)[scale=1.3,color=green!50!black]{$f(k) \cdot (n+m)$};
	\node at (-6.8*\x,-2.1*\y)[scale=1.3,color=orange!85!black]{$f(k) \cdot n^2$};
	\node at (-6.75*\x,-4.25*\y)[scale=1.3,color=red]{GP-hard};
	\node at (-5.9*\x,1.0*\y)[scale=1.3,color=blue,align=center]{enum-advice \\ kernel};

	\end{tikzpicture}
	\caption{
		``Layerwise separation'' of considered parameters with respect to known and new results. Herein, the parameters are hierarchically arranged in the sense that if two parameters are connected by a line, then the lower one can be upper-bounded by some function only depending on the higher one. Thus, hardness results transfer downwards and tractability results upwards.
		For the family of parameters distance to~$d$-degenerate graphs we highlighted~$d=0$ and~$d=1$ as prominent examples.
	}
	\label{fig:paramh}
\end{figure}
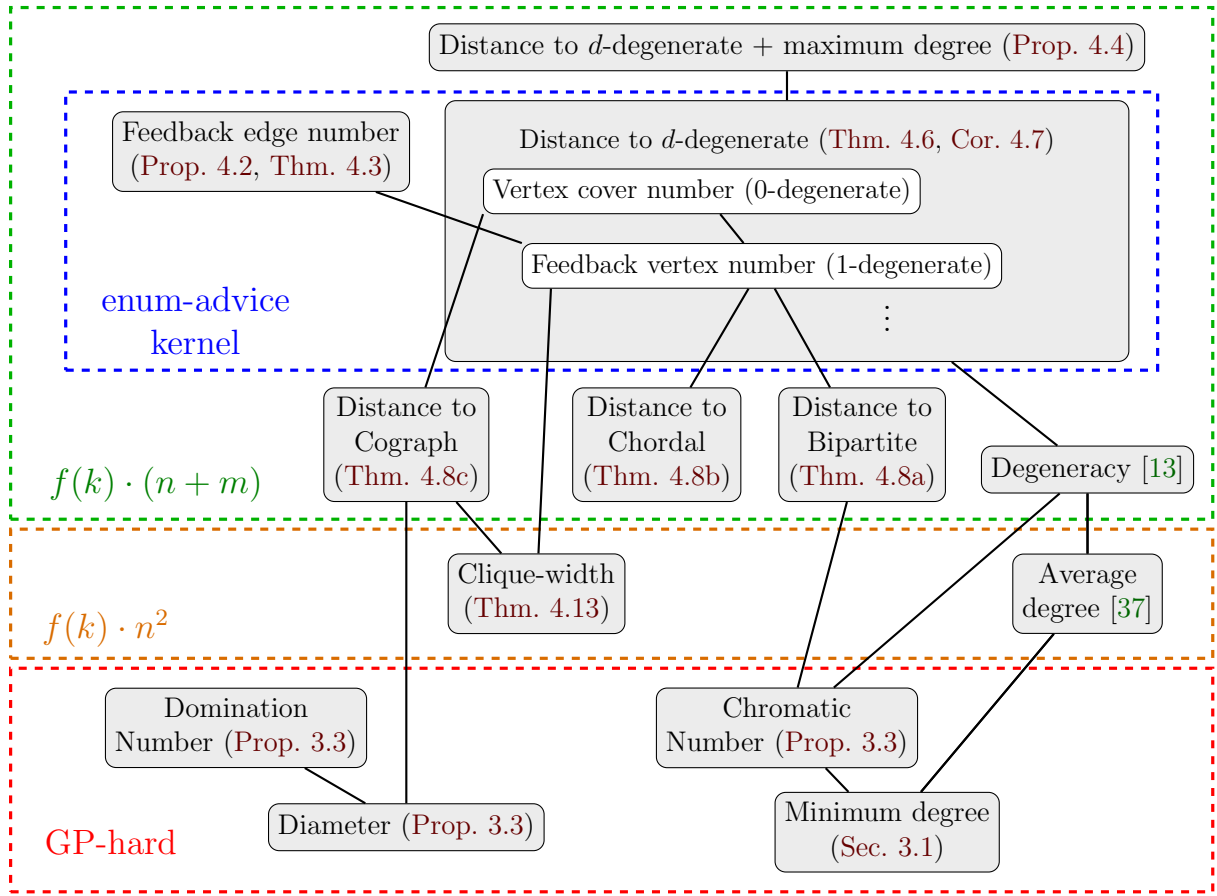
In this section, we present FPT algorithms solving \trenum{} exploiting several parameters.
We systematically explore the parameter space along a hierarchy of graph parameters (see~\cite{parahi}) in the following way (\cref{fig:paramh} surveys our line of attack).
We start from the fact that \mbox{\trenum{}} allows for an $O(m\cdot d)$-time algorithm when parameterized by degeneracy~$d$~\cite{DBLP:journals/siamcomp/ChibaN85}, and go into the following two directions: 
First, we study in \cref{sec:abovedeg} whether parameters upper-bounding degeneracy 
admit algorithms running in~${f(k) + O(n+m)}$ time, where~$k$ is the respective parameter value.
Kernelization is one way to achieve such \emph{additive}~($f(k) + O(n+m)$) instead of \emph{multiplicative} ($f(k)\cdot O(n+m)$) running times.
Indeed, we will present enum advice-kernelizations in this section.
Second, we study in \cref{sec:unrelatedtodeg} parameters that are
incomparable with degeneracy and so far have been unclassified.

We remark that in most of our running times the dependency on the parameter is modest. 
Thus, in scenarios where the respective parameter is small, the corresponding algorithms should be suitable for implementation.

\subsection{Parameters Lower-Bounded By Degeneracy}
\label{sec:abovedeg}
In this section we show results on graph parameters that upper-bound a graph's degeneracy.
In each subsection, we first describe the respective parameter and then turn to the results.

\subsubsection{Feedback Edge Number}
We begin with feedback edge number.
A \emph{feedback edge set} in a graph is a subset of the edges such that removing the edges from the graph results in a forest. 
The \emph{feedback edge number} of a graph is the size of a minimum feedback edge set.
Graphs with small feedback edge set number are ``almost trees''; such social networks occur in the context of sexually transmitted infections~\cite{PPM02} and extremism propagation~\cite{FNKS08}.
This parameter was recently used to achieve a significant speed-up in the computation of maximum matchings \cite{KNNZ18}.
The feedback edge number is neither related to the distance to $0$-degenerate graphs (vertex cover number) nor to the maximum degree, but it upper-bounds the distance to $1$-degenerate graphs (feedback vertex number).
Note that the feedback edge number is~$m-n+c$ where~$c$ is the number of connected components.
Hence the parameter can be of order~$O(m)$.
It can be computed in linear time by e.\,g.\ breadth-first-search.
We hence assume that an optimal feedback edge set is given.

We first provide a key lemma and then state a linear-size enum-advice kernel for \trenum{} parameterized by feedback edge number.
Recall that the feedback edge number of a graph is the size of a smallest subset of the edges such that removing the edges from the graph results in a forest.

\begin{lemma}
\label{lem:festri}
Let~$G=(V,E)$ be an undirected graph and let~$F$ be a feedback edge set in~$G$. All triangles~$\{u,v,w\}$ where at least one of the edges between the three vertices is not in~$F$ can be enumerated in~$O(n+m)$ time. There are at most~$2|F|$ such triangles.
\end{lemma}

\begin{proof}
Let~$G'=(V,E\setminus F)$ be an arbitrarily rooted forest and let~$p(v)$ denote the parent of a vertex~$v$ in it.
Note that~$F$,~$G'$, and~$p$ can be computed in~$O(n+m)$ time.
Observe further that all edges in~$E\setminus F$ are of the form~$\{v,p(v)\}$ for some vertex~$v$ and hence every triangle~$\{u,v,w\}$ in~$G$ where at least one of the edges between the three vertices is not in~$F$ is of the form~$\{u,v,p(v)\}$ for some vertices~$u,v$.
Note that there are at most two such triangles per edge in~$F$.
We can list all such triangles in linear time by the following algorithm.
We first mark all edges in~$F$ so that we can check for each edge in constant time whether it is in~$F$ or in~$E\setminus F$.
We first iterate over all vertices~$v\in V$ and find all triangles where exactly two edges between the three vertices are in~$F$ in overall linear time.
We iterate over all edges~${v,w}$ incident to~$v$ and if~$\{v,w\} \in F$, then we mark~$w$.
Afterwards, we iterate again over all neighbors~$w$ of~$v$ and if both~$w$ and~$p(w)$ are marked, then we list~$\{v,w,p(w)\}$ as a triangle.
In a third iteration we remove all markings from the neighbors of~$v$.
It remains to list all triangles with exactly one edge in~$F$.
To this end, we iterate over all edges~$\{u,v\} \in F$ and check whether~$p(u) = p(v)$,~$p(p(u)) = v$, or~$u = p(p(v))$ in constant time.
The algorithm takes linear time as~$O(\sum_{v\in V} \deg v) = O(m)$.

Assume towards a contradiction that there is a triangle~$\{x,y,z\}$ in~$G$ which is not listed by the described algorithm and where at least one edge between these three vertices is not in~$F$.
Without loss of generality let~$\{x,y\}\in E\setminus F$.
Since~$G'$ is a rooted forest, either~$x$ is the parent of~$y$ or~$y$ is the parent of~$x$.
Let without loss of generality be~$y=p(x)$.
Since~$\{x,y,z\}$ is a triangle, it holds that~$\{y,z\}\in E$.
If~$\{y,z\} \in F$, then~$\{x,y,z\}$ is listed when choosing~$v = x$ in the algorithm above.
By construction, both~$y$ and~$z$ are marked in the first iteration and then~$\{x,y,z\}$ is listed in the second iteration.
If~$\{y,z\} \in E\setminus F$, then~ either~$y = p(z)$ or~$p(y) = z$.
In the former case it holds that~$p(x) = y = p(z)$ and in the latter case it holds that~$z = p(y) = p(p(x))$ and $\{x,y,z\}$ is therefore listed.
\end{proof}

Using this we can easily show the following enum-advice kernel.

\begin{proposition}
	\label{fesnkernel}
	\trenum{} parameterized by feedback edge number~$k$ 
	admits a constant-delay enum-advice kernel with at most~$2k+3$ vertices and~${k+3}$ edges.
	It can be computed in~$O(n+m)$ time. 
\end{proposition}

\begin{proof}
Construct an enum-advice kernel~\mbox{$(I=I(G,k), A=A(G,k))$} as follows. 
For every edge~$e \in F$ put~$e$ and both of its endpoints into the graph.
Compute the feedback edge number~$k'\leq k$ of~$G_I$ in linear time.
Compute all triangles in~$G$ with at least one edge in~$E\setminus F$ and set~$A$ to be the set of all triangles found.
If~$A \neq \emptyset$, then add one extra triangle~$\{x,y,z\}$ where~$x,y,z \notin V$.
\cref{lem:festri} shows that there are at most~$2k$ such triangles and that they can be computed in linear time.
Observe that each step can be done in~$O(n+m)$ time.
Set the function~$f(\{x,y,z\},A)=A$ and~$f(w,A)=\{w\}$ for each~$w\in \operatorname{Sol}(I), w\neq \{x,y,z\}$.

We prove that the algorithm fulfills conditions~\condRef{K}{prop:eak-small}--\condRef{K}{prop:eak-enum} of~\cref{def:eakernel}.
  
By construction, for each edge in~$F$ there are at most two vertices and one edge put into~$I$.
There is at most one extra triangle added with three vertices and three edges.
Thus, it holds that~$|G_I| \leq~3\cdot k+3$~\condRef{K}{prop:eak-small}.

Assume that there is a triangle~$T = \{v_x,v_y,v_z\}$ in~$G$.
It either contains at least one edge in~$E\setminus F$ or only edges in~$F$.
In the first case~$G_I$ contains the triangle~$\{x,y,z\}$ and in the second case~$G_I$ contains~$T$.
Analogously, if~$G_I$ contains a triangle~$T'$, then it is either~$\{x,y,z\}$ or not.
If it is, then, by construction,~$A\neq \emptyset$ and hence~$G$ contains a triangle in~$A$.
If it is not, then~$T'$ is also contained in~$G$.
Thus~$G$ contains a triangle if and only if~$G_I$ contains a triangle \condRef{K}{prop:eak-correct}.

It remains to discuss the properties \condRef{K}{prop:eak-enum} of function~$f$.

For~$p,q\in \sol(I)\setminus \{x,y,z\}$, if~$p \neq q$, then~$f(p,A) \cap f(q, A) = \{p\} \cap \{q\} =\emptyset$.
If~$p$ or~$q$ is~$\{x,y,z\}$ (without loss of generality~$p = \{x,y,z\}$), then~$f(p,A)$ only contains triangles with at least one edge not in~$F$ and~$f(q,A)=\{q\}$ contains only a triangle where all edges are in~$F$.
It follows that~$f(p,A) \cap f(q,A) = \emptyset$ \condRef{K}{prop:eak-enum-disjoint}. 

By construction and by \cref{lem:festri},~$f(\{x,y,z\},A)$ contains all triangles in~$G$ where at least one of the edges is not in~$F$.
Since all edges in~$F$ are included in~$G_I$, all other triangles are contained in~$G_I$ \condRef{K}{prop:eak-enum-all}.
  
It is easy to see that~$f$ can be computed in constant-delay time \condRef{K}{prop:eak-enum-algo} by either iterating over~$A$ or just forwarding~$w$.
\end{proof}

To the best of our knowledge, there is no known algorithm that solves \trenum{} in~$O(n+m)$ or constant-delay time.

A straight-forward application of \cref{lem:kernel-alg} combined with \cref{fesnkernel} yields the following.

\begin{theorem}
\label{cor:fesnalg}
\trenum{} parameterized by feedback edge number~$k$ can be solved in $O(k^{1.5} + n + m)$ time.
\end{theorem}

\begin{proof}
  By~\cref{fesnkernel}, \trenum{} parameterized by feedback edge number~$k$ admits a constant-delay enum-advice kernel with at most~\mbox{$k$} edges.
  It can be computed in~$O(n + m)$ time and the size of the advice is in~$O(k)$. 
  Itai and Rodeh~\cite[Theorem 3]{DBLP:journals/siamcomp/ItaiR78} showed that the number of triangles in a  graph is upper-bounded by~$m^{3/2}$ and that all triangles in a graph can be enumerated in~$O(m^{3/2})$ time.
  Let~$(G,k)$ be an input instance of \trenum{} parameterized by feedback edge number and let~$I(G,k)$ be the kernel instance of the described enum-advice kernelization.
  Since~$|I(G,k)| \in O(k)$, all solutions can be listed in $O(|I(G,k)|^{3/2})$ time and the number of triangles in both instances is in~$O(k^{3/2})$.
  The statement of the theorem then follows directly from \cref{lem:kernel-alg}.
\end{proof}

\subsubsection{Distance to~$d$-Degenerate Graphs plus Maximum Degree}
We next turn to the parameters distance to~$d$-degenerate graphs and maximum degree.
A graph is \emph{$d$-de\-ge\-ne\-rate} if each induced subgraph contains a vertex of degree at most~$d$.
The \emph{distance to~$d$-de\-ge\-ne\-rate graphs} of a graph~$G$ is the size of a minimum-cardinality vertex set~$D$ such that~$G - D$ is~$d$-degenerate. 
This parameter generalizes several well-known parameters like vertex cover (distance to 0-degenerate graphs) and feedback vertex set (distance to 1-degenerate graphs).
For any fixed~$d$ the distance to~$d$-degenerate graphs is~$\NP$-hard to compute~\cite{DBLP:journals/jcss/LewisY80}.
However, we can use existing linear-time constant-factor approximation algorithms for~$d=0$ and~$d=1$~\cite{BGNR98}.
Since a minimum feedback vertex set of a graph is always (and possibly much) smaller than its smallest vertex cover, it is natural to use this parameter rather than the vertex cover number if comparably good results can be shown for both parameters.

For larger values of~$d$, one can use heuristics to compute small sets~$D$ such that~$G - D$ is~$d$-degenerate.
Notably, the quality of the heuristic only affects the running time but not the solution quality of the subsequent parameterized algorithm. 

The distance to~$d$-degenerate graphs is usually small in many applications such as social networks as they contain only few vertices with high degree~\cite{DBLP:reference/snam/Ferrara14}.
Depending on the degree distribution at hand one can then choose the value of~$d$ that gives the best overall running-time. 
(The running time of the corresponding algorithms usually has some trade-off between~$d$ and the distance to~$d$-degenerate graphs.)

Green and Bader~\cite{DBLP:conf/socialcom/GreenB13} stated that \textsc{Triangle Counting} parameterized by the size of a vertex cover~$V'$ and the maximum degree~$d_{\max} = \max(\{\deg(v) \mid v \in V'\})$ of vertices in this vertex cover can be solved in~$O(|V'| \cdot d_{\max}^2 + n)$ time. 
We will construct an algorithm which solves \trenum{} parameterized by the distance to~$d$-degenerate graphs, provided that the respective set is given. 
Let~$D$ be set of vertices such that~$G-D$ is $d$-degenerate, where~$G$ is the input graph.
Let~$\Delta_D$ be the maximum degree of a vertex in~$D$ with respect to~$G$.
Our algorithm takes~$O(|D| \cdot \Delta_D^2 + n \cdot d^2)$ time. 
Note that for each vertex cover~$V'$ it holds that~$G - V'$ is~$0$\nobreakdash-de\-gen\-er\-ate. 
Hence applying our algorithm with~$d=0$ yields a running time of~$O(|D| \cdot \Delta_D^2 + n \cdot d^2) = O(|D| \cdot \Delta_D^2 + n)$. 
This matches the running time of Green and Bader's algorithm.
Consequently, our result generalizes the result by Green and Bader.
\begin{proposition}
	\label{Bader}
	\trenum{} parameterized by distance to $d$-de\-gen\-er\-ate graphs and maximum degree~$\Delta_D$ in a set~$D$ such that~$G-D$ is~$d$\nobreakdash-de\-gen\-er\-ate is solvable in~${O(|D|\cdot\Delta_D^2+n\cdot d^2)}$ time provided that the set~$D$ is given.
\end{proposition}
\begin{proof}
	Let~$D$ be a set such that~$G - D$ is~$d$-degenerate and let the maximum degree in~$D$ be~$\Delta_D = \max(\{\deg_G(v) \mid v \in D\})$.
	We show how to list all triangles in~$G$ in~${O(|D| \cdot \Delta_D^2+n\cdot d^2)}$ time in two steps.

	In the first step, list all triangles which do not contain any vertices in~$D$. 
	To this end, compute the~$d$-degenerate graph~$G' = G - D$ and list all triangles contained in~$G'$ in~$O(n \cdot d^2)$ time~\cite{DBLP:journals/siamcomp/ChibaN85}.

	In the second step, list all triangles with at least one vertex contained in~$D$. 
	To this end, in linear time fix an arbitrary strict order~$<_a$ on~$V(G)$ such that~\mbox{$v<_a w$} for all~$v\in D, w\in V(G')$. 
	For each~$u \in D$, iterate over all of the at most~$\Delta_D^2$ possible pairs of neighbors~$v,w\in N(u)$. 
	For each pair~$v,w\in N(u)$, check in constant time whether (i)~$\{v,w\}\in~E(G)$ and (ii)~$u<_a v<_a w$, and list the triangle if both conditions are met.
	Let~$\{x,y,z\}$ form a triangle in~$G$ with at least one vertex in~$D$, and without loss of generality let~$x<_a y <_a z$.
	Then~$x\in D$, and only in the iteration when~$x$ is chosen from~$D$ conditions~(i) and~(ii) are met and hence~$\{x,y,z\}$ is listed.
	It follows that in the second step all triangles with at least one vertex in~$D$ are enumerated exactly once in~$O(|D| \cdot \Delta_D^2)$~time.
\end{proof}

Using the above ideas, we also provide an enum-advice kernel for \trenum{} parameterized by distance to~$d$-degenerate graphs and the maximum degree in the deletion set~$D$. 

\begin{observation}
\label{obs:dmkernel}
	\trenum{} parameterized by distance to $d$-de\-gen\-er\-ate graphs and maximum degree~$\Delta_D$ in a set~$D$ such that~$G-D$ is~$d$\nobreakdash-de\-gen\-er\-ate admits a con\-stant-delay enum-advice kernel provided that the distance to $d$-degenerate graphs deletion set~$D$ is given. 
	The kernel is of size $O(|D| \cdot \Delta_D \cdot d)$ and can be computed in~$O(n\cdot d^2 + |D| \cdot \Delta_D)$~time. 
\end{observation}

\begin{proof}
	The main idea is to only compute the first step of the algorithm in the proof of \cref{Bader}. 
	Store in the advice all the triangles that are found in the process such that at least one vertex of the triangle has no neighbor in~$D$. 
	(In this way, we avoid double counting by having a triangle in the advice and in the kernel.)
	This first step takes~$O(n\cdot d^2)$ time. 
	Then, the kernel contains the subgraph induced by~$D$ and all neighbors of vertices in~$D$. 
	The kernel can be computed in~$O(n + m) \subseteq O(n\cdot d + |D| \cdot \Delta_D)$ time and the resulting graph contains at most~$|D| \cdot (\Delta_D + 1)$ vertices, at most~$|D| \cdot \Delta_D$ edges with at least one endpoint in~$D$, and at most~$|D| \cdot \Delta_D \cdot d$ edges with no endpoint in~$D$ since~$G-D$ is $d$-degenerate.
\end{proof}

\subsubsection{Distance to~$d$-Degenerate Graphs}
We next present an enum-advice kernel for \trenum{} parameterized by distance to~$d$-degenerate graphs. 
Recall that the distance to~$d$-de\-ge\-ne\-rate graphs of a graph~$G$ is the size of a minimum-cardinality vertex set~$D$ such that each induced subgraph of~$G - D$  contains a vertex of degree at most~$d$ (that is,~$G-D$ is~$d$-degenerate).
The ideas for the kernel regarding the distance to~$d$-degenerate graphs are a little bit different than the ones for \cref{obs:dmkernel}; we will, however, start similarly and enumerate all triangles in the~$d$-degenerate subgraph and store them in the advice.

\begin{theorem}
	\label{thm:dtddkernel}
	\trenum{} parameterized by distance to~$d$-de\-gen\-er\-ate graphs admits a con\-stant-delay enum-advice kernel provided that the distance deletion set~$D$ to $d$-degenerate graphs is given. 
	The kernel contains at most $|D| + 2^{|D|} +3$ vertices and can be computed \mbox{in~$O(n \cdot (d+1) \cdot (|D| + d))$}~time. 
\end{theorem}

\begin{proof}
Let~$G$ be an instance of \trenum{} and let~\mbox{$k=|D|$} be the parameter. 
Construct the enum-advice kernel~$(I(G,k), A(G,k))$ as follows. 
To this end, we call~${(G_I,k')\defeq I(G,k)}$ and~$A\defeq A(G,k)$.

First, compute the~$d$-degenerate graph~$G' = G - D$.
The graph~$G'$ contains exactly those triangles in~$G$ that do not contain any vertices in~$D$. 
  Using a result of Chiba and Nishizeki~\cite[Theorem 1]{DBLP:journals/siamcomp/ChibaN85}, compute the set of triangles in~$G'$ in~$O(m \cdot d)$ time.
Next, compute all triangles with exactly one vertex in~$D$. 
To this end, compute the degeneracy ordering in linear time~\cite{MatulaB83}, iterate over all~\mbox{$v \in D$,~$u \in N(v) \setminus D$}, and the at most~$d$ neighbors of~$u$ that are ordered after~$u$ in the degeneracy order, and list all triangles found.
By this, all triangles in~$G$ containing exactly one vertex in~$D$ are found 
in~$O(k \cdot n \cdot d)$ time.
Altogether, we can compute the set~$T_1$ of all triangles in~$G$ with at most one vertex 
in~$D$ in~$O(n \cdot d \cdot (k + d))$~time. 

Delete all edges which have no endpoint in~$D$ as they cannot be part of any further 
triangles.
Next, compute all modules in the current graph, that is, a partition~$\mathcal{P}$ of the vertices according to their neighbors, using partition refinement in~$O(n+m)$~time~\cite{DBLP:conf/stacs/HabibPV98}.

For each non-empty part~$P \in \mathcal{P}$ pick one vertex~$v_P \in P$ and store 
a function~$M$ such that~$M(v_P) = P \setminus D$. 
Put all vertices in~$D$, all of the chosen vertices, and all edges induced by these 
vertices into~$G_I$. 
Add three new vertices~$a,b,c$ to~$G_I$ and if~${T_1 \neq \emptyset}$, then add 
three new edges~$\{a,b\},\{a,c\},\{b,c\}$. 
Note that all edges have an endpoint in~$D'= D \cup \{a,b\}$ and thus~$D'$ is a deletion set to~$d$-degenerate graphs for every~$d$.
Complete the construction by setting~$k'= |D'|$ and~$A = (T_1,M,\{a,b,c\})$.
Note that~$G_I$ contains at most~$k+2^{k}+3$ vertices~\condRef{K}{prop:eak-small}.
Observe that since~${m \in O(n \cdot (k + d))}$, the kernel can be constructed in~${O(n \cdot d \cdot (k + d))}$ time. 
For~$x_1,x_2,x_3\in V(G_I)$, define the function~$f$ as~$$f(\{x_1,x_2,x_3\},A) \defeq \begin{cases} T_1  \text{, if } \{x_1,x_2,x_3\} = \{a,b,c\} \text{, and otherwise} \\ \{\{v_1,v_2,v_3\} \mid v_1 \in M(x_1) \land v_2 \in M(x_2) \land v_3 \in M(x_3)\}. \end{cases}$$
Next, we prove that the algorithm fulfills all conditions of \cref{def:eakernel}.

Observe that~$G_I$ is isomorphic to a subgraph of~$G$ and, hence, if there is a 
triangle~$G_I$, then there is a triangle in~$G$.
Assume that there is a triangle~$X$ with vertices~$\{x_1,x_2,x_3\}$ in~$G$.
If~$X$ contains at most one vertex in~$D$, then~$T_1 \neq \emptyset$ and thus 
there is the triangle formed by~$\{a,b,c\}$ in~$G_I$. 
Otherwise,~$X$ contains at least two vertices in~$D$. Assume without loss of 
generality that~$x_2, x_3 \in D$. 
If~$x_1$ is in~$D$, then~$X$ is also contained in~$G_I$. 
Otherwise, there is a vertex~$v$ in~$G_I$ such that~$x_1 \in M(v)$. 
Since~$\{x_1,x_2,x_3\}$ forms a triangle in~$G$, it follows that~$\{v,x_2,x_3\}$ forms 
a triangle in~$G$ and~$G_I$.
Hence, condition~\condRef{K}{prop:eak-correct} (of \cref{def:eakernel}) is fulfilled.

Next we discuss the condition~\condRef{K}{prop:eak-enum}.
We will prove that for each triangle~${X = \{x_1,x_2,x_3\}}$ in~$G$ there is a unique solution~$w \in \sol(G_I,k')$ such that~$X \in f(w, A)$ \condRef{K}{prop:eak-enum-all}. 
If~$X$ contains at most one vertex in~$D$, then by construction~$X \in f(\{a,b,c\}, A)$. 
Since~$G_I$ contains only edges with an endpoint in~$D$, no triangle~$\{v_1, v_2, \allowbreak v_3\}$ where~$v_1 \in M(x_1)$,~$v_2 \in M(x_2)$, and~${v_3 \in M(x_3)}$ is contained in~$G_I$. 
Thus,~$\{a,b,c\}$ is the only triangle~$T$ such that~${X \in f(T, A)}$.
If~$X$ contains at least two vertices~$x_2,x_3\in D$, then there exists a vertex~$v$ in~$G_I$ such that~$x_1 \in M(v)$ and the triangle~$\{v,x_2,x_3\}$ is contained in~$G$.
By construction, the triangle~$\{v,x_2,x_3\}$ is also contained in~$G_I$
and~$X \in f(\{v,x_2,x_3\}, A)$. 
Since~$X \notin T_1$, it follows~$X \notin f(\{a,b,c\}, A)$. 

Next we show that for any two triangles~$p = \{u_1,u_2,u_3\}$ and~$q = \{v_1,v_2,v_3\}$ 
in~$G_I$, it holds that~$f(p,A)\cap f(q,A)=\emptyset$ \condRef{K}{prop:eak-enum-disjoint}.

If either~$p$ or~$q$ is~$\{a,b,c\}$ (without loss of generality~$p = \{a,b,c\}$), 
then by definition~$f(p,A)$ only contains triangles with at most one vertex in~$D$ 
and~$f(q, A)$ only contains triangles with at least two vertices in~$D$ and 
thus ${f(p,A) \cap f(q,A) = \emptyset}$.

If neither~$p$ nor~$q$ is~$\{a,b,c\}$, then both of them only contain vertices 
from the original graph~$G$. 
As~$p\neq q$, assume without loss of generality that~$u_1 \notin q$ and~\mbox{$v_1 \notin p$}. 
By construction all triangles in~$f(p, A)$ contain one vertex in~$M(u_1)$ and all 
triangles in~$f(q, A)$ contain one vertex in~$M(v_1)$. 
As shown above,~$M(u_1)$ ($M(v_1)$, respectively) only contains~$u_1$~($v_1$) and vertices that 
have the same neighbors as~$u_1$~($v_1$) in~$D$. 
Hence, no triangle in~$f(p, A)$~($f(q, A)$ respectively) contains a vertex 
in~$M(v_1)$ ($M(u_1)$) and thus~$f(p, A) \cap f(q, A) = \emptyset$.

Each triangle in~$\{\{v_1,v_2,v_3\} \mid v_1 \in M(x_1) \land v_2 \in M(x_2) \land v_3 \in M(x_3)\}$ and in~$T_1$ can be returned with constant delay between generating two successive solutions~\condRef{K}{prop:eak-enum-algo}.

Overall, the time needed to compute the kernel~$(I(G,|D|),A(G,|D|))$ is upper-bounded by~$O(n \cdot d \cdot (|D| + d) + |D| + m) = O(n \cdot (d + 1) \cdot (|D| + d))$. 
The equality holds since~$m \in O(n \cdot (|D| + d))$.
\end{proof}

To the best of our knowledge, there is no algorithm that solves \trenum{} parameterized by distance to~$d$-degenerate graphs within $O(n\cdot d^2 + |D| \cdot \Delta_D)$~time. 
All solutions can be reconstructed in constant-delay time and there is no known algorithm that solves \trenum{} in constant-delay time (and it seems unlikely that such an algorithm exists).

Using \cref{lem:kernel-alg} and \cref{thm:dtddkernel} we get the following result.

\begin{corollary}
	\trenum{} parameterized by distance to~$d$-de\-gen\-er\-ate graphs is solvable in~$O(n \cdot (d+1) \cdot (|D| + d) + 2^{3|D|} + \hashT)$ time provided that the vertex-deletion set~$D$ to $d$-degenerate graphs is given.
	\label{cor:dtddalg}
\end{corollary}

\begin{proof}
By \cref{thm:dtddkernel}, \trenum{} parameterized by the distance
to~$d$-de\-ge\-ne\-rate graphs (provided that the set~$D$ such that~$G-D$ is~$d$-de\-ge\-ne\-rate is given)
admits a constant-delay enum-advice kernel with size~$O(2^{2|D|})$ that can be computed 
in~$O(n \cdot (d+1) \cdot (|D|+d))$ time. 
Hence, all triangles in the kernel instance can be found in~$O(2^{3|D|})$ time \cite{DBLP:journals/siamcomp/ItaiR78}. Since the delay is constant and the number of triangles in both graphs is at most~$\hashT$, we can compute all triangles in the original instance from all solutions in the kernel in~$O(\hashT)$ time. Thus, by \cref{lem:kernel-alg}, \trenum{} is solvable in~${O(n \cdot (d+1) \cdot (|D| + d) + 2^{3|D|} + \hashT)}$ time parameterized by distance to~$d$-degenerate graphs assuming that the set~$D$ is given.
\end{proof}

\subsection{Parameters Incomparable with Degeneracy}
\label{sec:unrelatedtodeg}
In this section we present results on parameters that are unrelated to the degeneracy.
Again, we first describe the parameters and then turn to our results.

In \cref{sec:distance}, we consider the vertex-deletion distance to cographs, bipartite, or chordal graphs.
A graph is bipartite if its set of vertices can be partitioned into two sets such that no edge in the graph has both endpoints in one of the sets. 
A graph is called chordal if it does not contain induced cycles of length at least four.  
A graph is called a cograph if it contains no induced path with four vertices~($P_4$).

We show below that enumerating all triangles is easy if the input graph falls into one of the three graph classes. 
Thus, the three parameters mentioned above are natural candidates for a ``distance-to-triviality'' approach~\cite{GHN04}. 
Furthermore, all three parameters are upper-bounded by the vertex cover number. The vertex cover number allows for tractability results (see \cref{sec:abovedeg}).
Thus, aiming at generalizing the tractability result, we arrive at the study of these parameters.
Moreover, distance to bipartite graphs and distance to cographs are lower-bounded by parameters for which we know intractability, see \cref{fig:paramh}.
Thus, we also investigate the limits of how far we can generalize the tractability results.

Distance to cographs lower-bounds the cluster vertex number---a parameter advocated by Doucha and Kratochv{\'{\i}}l~\cite{DK12} by proving that several basic graph-theoretic problems parameterized by it are fixed-parameter tractable (note that it lies between the vertex cover number and clique-width).
Moreover, given a graph~$G$, we can determine in linear time whether~$G$ is a cograph and, if this is not the case, return an induced~$P_4$~\cite{journals/BretscherCHP08,DBLP:journals/siamcomp/CorneilPS85}. 
This implies that in~$O(k\cdot (m+n))$ time, with~$k$ being the distance to cographs, we can compute a set~$K\subseteq V$ of size at most~$4k$ such that~$G-K$ is a cograph.
However, we are not aware of (parameterized) linear-time constant-factor approximation algorithms for distance to bipartite or distance to chordal.

In \cref{sec:cliquewidth}, we consider the parameter clique-width.
Since treewidth is lower-bounded by degeneracy, we know that there is an $O(\tau\cdot m)$-time algorithm for \trenum{}, where~$\tau$ is the treewidth of the input graph.
A parameter lower bounding tree\-width in the parameter hierarchy is clique-width~$k$ (it holds that~$k \le 2^{\tau+1}+1$ and~$k$ can be arbitrarily small compared to~$\tau$~\cite{CourcelleO00}).
Moreover, clique-width also lower-bounds distance to cograph.
Thus, we study clique-width as it lies on the ``border to tractability'' of \trenum{}.

\subsubsection{Distance to Bipartite Graphs, Chordal Graphs, or Cograph}
\label{sec:distance}
We give linear-time FPT algorithms for \trenum{} with respect to the distance to bipartite, distance to chordal, and distance to cographs, respectively.
Our main results in this section are summarized in the following.

\begin{theorem}
 \label{thm:distbipchorcog}
 \trenum{} is solvable in $O(n+m \log n \cdot |K|+M)$ time
 \begin{enumerate}[(a)]
  \item with~$M=0$, when parameterized by the distance~$k$ to bipartite graphs, provided that the deletion set is given;\label{thm:distbip}
  \item with $M=\hashT$, when parameterized by the distance~$k$ to chordal graphs, provided that the deletion set is given;\label{thm:distchor}
  \item with $M=\hashT$, when parameterized by the distance~$k$ to cographs.\label{thm:distcog}
 \end{enumerate}
\end{theorem}

In order to prove~\cref{thm:distbipchorcog}, we provide a general lemma which can be used to solve~\trenum{} with a given (vertex) deletion set to some graph class~$\Pi$ if all triangles in~$\Pi$ can be enumerated efficiently. 

\begin{lemma}
	\label{Pi}
	Let~$\Pi$ be some graph class and let~$f_\Pi(n,m)$ be the time required to solve \trenum{} on graphs in~$\Pi$.
	Then,	\trenum{} with a given vertex-deletion set~$K$ to $\Pi$ is solvable in~$O(m \cdot |K| + n+ f_\Pi(n,m))$ time.
\end{lemma}

\begin{proof} 
Let $K$ be a set of vertices such that $G' = G - K$ is a graph contained in~$\Pi$. 
By definition, all triangles within $G'$ can be listed in $O(f_\Pi(n,m))$ time. 
All triangles with at least one vertex in $K$ can be listed in $O(m \cdot |K| + n)$ time by the following algorithm. 
Read the whole input and fix an arbitrary linear order $\leq_a$ of the vertices in $K$ in~$O(n+m)$ time. 
Check for each edge~$\{u,w\}\in E(G)$ and each vertex~$v\in K$ whether~$\{u,v,w\}$ is a triangle and for all~$x\in \{u,w\} \cap K$ it holds that $v \leq_a x$.
This can be done for all edges and one vertex~$v\in K$ in~$O(m)$ time by first marking all neighbors of~$v$, then check for each edge whether both endpoints are marked and finally remove all markings (by again iterating over all neighbors of~$v$).
If both conditions hold, then list $\{u,v,w\}$ as a new triangle. 
We prove that this algorithm lists all triangles with at least one vertex in $K$ exactly once. 
Since~$v \in K$ holds, this algorithm does not list any triangles which do not contain vertices in $K$. 
Let $\{a,b,c\}$ be an arbitrary triangle and let $a$ be in $K$. 
This triangle is found at least once as $\{b,c\}\in E$ and $a \in K$ holds. 
If for all~$x \in \{b,c\}$ it holds that~$x\notin K$ or~$a \leq_a x$, then this triangle is listed in the iteration where~$v = a$ and~$\{u,w\} = \{b,c\}$. 
Otherwise, $b \leq_a a$ or~$c \leq_a a$ holds. 
Without loss of generality, let~$b$ be such that~$b\leq_a y$ with~$y\in \{a,b,c\}\cap K$.
Then $\{a,b,c\}$ is listed in the iteration where $v = b$ and $\{u,w\} = \{a,c\}$ holds. 
There are~$m \cdot |K|$ iterations and each iteration takes~$\log n$ time. 
Thus, \trenum{} parameterized by deletion set to $\Pi$ is solvable in $O(m \cdot |K| + n)$ time.
\end{proof} 

\cref{thm:distbipchorcog} follows immediately from applying~\cref{Pi} with $\Pi$ being the class of bipartite graphs, of chordal graphs, and of cographs.
To this end, in the remainder of this section, we provide the remaining requirements to apply \cref{Pi}, that is, we give the running times in which~\trenum{} is solvable on bipartite graphs (\cref{bipartite}), chordal graphs~(\cref{chordal}), and cographs~(\cref{ch}). 

\paragraph{Distance to Bipartite Graphs}
Since bipartite graphs do not contain cycles of odd length and thus are triangle-free, \trenum{} is solvable in constant time on bipartite graphs after reading the input. 

\begin{observation}
	\trenum{} is solvable in $O(n+m)$ time on bipartite graphs.
	\label{bipartite}
\end{observation}


\paragraph{Distance to Chordal Graphs}
For each chordal graph there is a perfect elimination order~$\leq_p$ of the vertices which can be computed in linear time. 
That is, for each vertex~$v$ all neighbors~$w$ of~$v$ with~$v \leq_p w$ form a clique.
We will use this to find all triangles in chordal graphs in~$O(\hashT{} +\, n + m)$ time and thus by \cref{Pi} solve \trenum{} given a deletion set~$K$ to chordal graphs in~$O(\hashT{} +\, n + m \log n \cdot |K|)$ time. Recall that~$\hashT{}$ is the number of triangles in~$G$. 
Note that a clique containing~$n$ vertices contains~$O(n^3)$~triangles and that graphs consisting of only one large clique are chordal. 
We therefore cannot avoid the term~$\hashT{}$ in the running time. 

\begin{proposition} \label{chordal}
	\trenum{} is solvable in~${O(\hashT{}+\,n+m)}$ time on chordal graphs.
	\end{proposition}
\begin{proof}
Compute a perfect elimination ordering in~$O(n + m)$ time. 
Next, list all triangles containing the first vertex in this ordering and delete it afterwards.
Proceed in this manner until no vertex is left.

Let $v$ be the vertex at the first position in the perfect elimination ordering in some iteration.
Listing all triangles containing~$v$ can be done as follows. 
As~$v$ is the first vertex in the ordering, there are no vertices before~$v$ and, by the definition of a perfect elimination ordering,~$N(v)$ forms a clique. 
Hence,~$v$ combined with any two of its neighbors forms a triangle. 
Thus, we list all triangles of the form~$\{v,x,y\}$ with~$x \in N(v)$, $y \in N(v)$, and~$x \neq y$. 
Once all of these triangles are listed,~$v$ is not contained in any triangle not being listed and hence one can delete it.
\end{proof}

\paragraph{Distance to Cograph}
We now show how to enumerate all triangles in a cograph.
To this end, we need the following notation.
Every cograph has a binary cotree representation which can be computed in linear time~\cite{DBLP:journals/siamcomp/CorneilPS85}. 
A cotree is a rooted tree in which each leaf corresponds to a vertex in the cograph and each inner node either represents a disjoint union or a join of its children. 
A join of two graphs~$(V_1, E_1), (V_2, E_2)$ with $V_1 \cap V_2 = \emptyset$ is the graph~$(V_1 \cup V_2, E_1 \cup E_2 \cup \{\{x,y\}\mid x\in V_1 \land y \in V_2\})$. 
We will use these representations to find all triangles in cographs in~$O(\hashT{} + n + m)$ time, where~$\hashT{}$~is the number of triangles in $G$. Note that one can compute a set~$K$ of size at most~$4k$ such that~$G - K$ is a cograph, where~$k$ is the size of a minimum set~$K'$ such that~$G-K'$ is a cograph, in~$O(k \cdot (n+m))$~time.

\begin{proposition}
	\label{ch}
	\trenum{} is solvable in~$O(\hashT{} + n + m)$ time on cographs.
\end{proposition}

\begin{proof}
Consider a dynamic program which stores for each node $p$ in the cotree all vertices~$V(p)$, all edges~$E(p)$ and all triangles~$T(p)$ in the corresponding subgraph of $G$. This can be done as follows:

	Let $q_1,q_2$ be the children of an inner node $p$ in the cotree.
	\begin{itemize}
	\item A single leaf node has one vertex and no edges or triangles.
	\item A union node has vertices $V(q_1) \cup V(q_2)$, edges $E(q_1) \cup E(q_2)$, and triangles~\mbox{$T(q_1) \cup T(q_2)$}.
	\item A join node has
	\begin{align*}
	V(p) ={ }& V(q_1) \cup V(q_2),\\
	E(p) ={ }& E(q_1) \cup E(q_2) \cup \{\{x,y\} \mid x \in V(q_1) \land y \in V(q_2)\} \text{, and}\\
	T(p) ={ }& T(q_1) \cup T(q_2) \cup \{\{x,y,z\} \mid x \in V(q_1) \land \{y,z\} \in E(q_2)\} \cup\\&  \{\{x,y,z\} \mid x \in V(q_2) \land \{y,z\} \in E(q_1)\}.
	\end{align*}

	\end{itemize}
	That is, a join node contains all edges the two children contain and all possible edges between vertices of them. A join node contains all triangles its two child-nodes contain and one triangle for each edge $\{y,z\}$ of one of its children and a vertex $x$ of the other, because edges $\{x,y\}$ and $\{x,z\}$ are in $E$ and therefore $\{x,y,z\}$ is a triangle.

	We will first prove that all triangles are enumerated that way and afterwards we will analyze the running time of the dynamic program.

	Let $\{a,b,c\}$ be any triangle in the cograph. We will prove that there is at least one node $p$ in the cotree with $\{a,b,c\} \in T(p)$. As each inner node keeps the triangles from its children, it follows that $\{a,b,c\} \in T(r)$ when $r$ is the root node of the cotree. Let without loss of generality be $p$ the least common ancestor of $a,b,$ and $c$, and let $q_1, q_2$ be the two children of $p$. As neither $\{a,b,c\} \in V(q_1)$ nor $\{a,b,c\} \in V(q_2)$, let us assume without loss of generality that $a \in V(q_1)$ and $b,c\in V(q_2)$. It holds that $\{b,c\}\in E(q_2)$ because there is an edge between $b$ and $c$ and they are both descendants of $q_2$. The node $p$ has to be a join node as $\{a,b\}, \{a,c\}\in E$ and $p$ is by definition the least common ancestor. By definition it holds that $\{\{x,y,z\} \mid x \in V(q_1) \land \{y,z\} \in E(q_2)\} \subseteq T(p)$. It follows that~$\{a,b,c\} \in T(p)$. Note that $\{a,b,c\}$ is only computed once in the least common ancestor node $p$ and then passed to the parent node. Hence,~$T(r)$ only contains~$\{a,b,c\}$ once and thus listing all triangles in $T(r)$ solves \trenum{}. 

	We will now analyze the running time. There are $n$ leaf nodes in the cotree each of which require a constant amount of time to compute. There are at most $n-1$ union nodes each of which only require a constant amount of time as they only need to point on their children's values. There are at most $n-1$ join nodes. Each edge and triangle is only added once and all other values do not need to be recomputed. A pointer to the edges and triangles in the child nodes is enough and only requires a constant amount of time to be set. Altogether, the global running time of this algorithm is in $O(\hashT{} + n + m)$. 
%
\end{proof}
\subsubsection{Clique-width}
\label{sec:cliquewidth}
We next turn to the parameter clique-width as it is incomparable to the degeneracy and upper-bounded by two parameters allowing for linear-time FPT algorithms: the distance to cographs (\cref{ch}) and treewidth (as treewidth upper-bounds the degeneracy).

The clique-width of a graph~$G$ is the minimum number~$k$ such that~$G$ can be constructed using a~$k$-expression. 
A~$k$-expression consists of four operations which use~$k$ labels \cite{DBLP:journals/mst/CourcelleMR00}. The operations are the following.
\begin{itemize}
	\item Creating a new vertex with some label~$i$.
	\item Disjoint union of two labeled graphs.
	\item Edge insertion between every vertex with label~$i$ to every vertex with label~$j$ for some labels~$i \neq j$.
	\item Renaming of label~$i$ to~$j$ for some~$i,j$.
\end{itemize}

Let~$V_1 \subseteq V$ be a set of vertices. 
A~\emph{twin class} in~$V_1$ is a set~$V' \subseteq V_1$ of vertices such that every vertex in~$V \setminus V_1$ either has all vertices in~$V'$ as its neighbors or none of them.
The set~$V_1$ is called an~\emph{$\ell$-module} if it can be partitioned into at most~$\ell$ twin classes. 
Let~$B$ be a rooted full binary tree whose leaves are in bijection to the vertices in~$V$. 
For each inner node~$p$ in~$B$ let~$V_p \subseteq V$ be the set of all vertices in~$G$ whose corresponding nodes in~$B$ are in the induced subtree of~$B$ rooted at~$p$. 
If for each inner node~$p$ in~$B$ the set~$V_p$ is an~$\ell$-module, then~$B$ is an~\emph{$\ell$-module decomposition}.
We will use this~$\ell$-module decomposition to construct a dynamic program to solve \trenum{} parameterized by some~$k$-expression in~${O(n^2 + k^2 \cdot n + \hashT{})}$ time.

We leave open whether \trenum{} parameterized by clique-width~$k$ admits a linear FPT algorithm.
Our results suggest that the parameters clique-width and average degree form the border case between parameters admitting linear FPT algorithms and those that are GP-hard.

\begin{theorem}
	\trenum{} parameterized by clique-width is solvable in~$O(n^2 + n \cdot k^2 + \hashT{})$ time, provided that a~$k$-expression of the input graph is given.
\label{thm:cliquewidth}
\end{theorem}

\begin{proof}
	Recall that~$\hashT{}$ is the number of triangles in~$G$. 
	Bui-Xuan et al.~\cite[Lemma 3.2]{DBLP:journals/ejc/Bui-XuanSTV13} proved the following:
	First, given a~$k$-expression tree~$B$ of~$G$ one can compute in overall~$O(n^2)$ time for every node~$u$ in~$B$ the partition of~$V_u$ into its twin classes~$Q_u(1),\ldots, Q_u(h_u)$ where~$V_u$ is the set of vertices corresponding to the leaves of the subtree of~$B$ rooted at~$u$. 
	Second, the maximum number~$h_u$ of twin classes for each node~$u$ in~$B$ is at most~$k$. 
	Third,~$B$ can be modified such that it becomes a full binary tree and thus combined with the twin classes becomes a~$k$-module decomposition of~$G$. 
	Fourth, the~$k$-module decomposition has only a single twin class in the root node and each twin class of a node~$u$ in~$B$ is fully contained in one of the twin classes of the parent~$v$ of~$u$. 
	We use these statements in our algorithm.
	
	We next describe which information is stored in our dynamic program.
	First, store the information that all vertices in~$Q_u(i)$ are contained in the twin class~$Q_v(b)$ by adding~$i$ to a set~$M_{u,b}$. 

	Next, for each node~$v$~in~$B$ store~$Q_v(1), \ldots, Q_v(h_v)$, the twin classes of~$v$ (which are already computed), the set~$E^v_{i,j}$ of all edges between vertices in twin classes~$Q_v(i)$ and~$Q_v(j)$~$(1\leq i,j \leq h_v)$, and the set~$LT^v$ of all triangles formed by vertices in~$V_v$. 
	Denote by~$F_{u,a,w,b}$ the set of all edges with one endpoint in~$Q_u(a)$ and one in~$Q_w(b)$ where~$u$ and~$w$ have the same parent in~$B$, formally,~$F_{u,a,w,b} = \{\{x,y\} \mid x \in Q_u(a) \land y\in Q_w(b)\} \cap E$.
	
	Note that by the definition of twin classes, either all vertices of the twin classes of two nodes with the same parent in~$B$ are pairwise connected or none of them are.

	The dynamic program is defined as follows. 
	A leaf node~$v$ in~$B$ has an empty set of triangles and only empty edge sets. 
	An inner node~$v$ with children~$u,w$ has for each~${1 \leq i,j \leq h_v}$ an edge set 
	\begin{align*}
		E^v_{i,j} = {} & \bigcup_{l\in M_{u,i}}\bigcup_{m\in M_{u,j}}E^u_{l,m} \cup \bigcup_{l\in M_{w,i}}\bigcup_{m\in M_{w,j}}E^w_{l,m} \cup \\
					{} & \bigcup_{m\in M_{u,i}} \bigcup_{l\in M_{w,j}} F_{u,m,w,l} \cup \bigcup_{m\in M_{u,j}} \bigcup_{l\in M_{w,i}} F_{u,m,w,l}
	\end{align*}
	and a set of triangles 
	$$LT^v = LT^u \cup LT^w \cup LT_{u,u,w} \cup LT_{w,w,u}$$ 
	where 
	\begin{multline*}
		LT_{x,x,y} = \bigcup_{o=1}^{h_x} \bigcup_{p=o}^{h_x} \bigcup_{q=1}^{h_y} \{\{a,b,c\}\mid a\in Q_x(o) \land b\in Q_x(p) \land c\in Q_y(q)~\land \\ 
			F_{x,o,y,q} \neq \emptyset \land F_{x,p,y,q} \neq \emptyset \land \{a,b\}\in E^x_{o,p}\}.
	\end{multline*}

	We next analyze the running time and then the correctness of the dynamic program.
	The table entries for each of the~$n$ leaves in~$B$ can be computed in constant time.
	For each of the~$n-1$ inner nodes, at most~$k^2$ sets of edges have to be computed, each of which is formed out of two parts: edges already stored in the children and edges between vertices of different children. 
	The former requires~$O(k^2)$ time per node as it can be seen as a list of pointers to the children's sets of edges. 
	The latter requires~$O(m)$ global time as each edge is only added once. 
	Hence, the edge sets of all nodes can be computed in~$O(n \cdot k^2 + m)$ time.
	The list of triangles is a list containing two pointers to its children's list of triangles and a third list of new triangles. 
	As each triangle is only added once, all lists can be computed in~$O(\hashT + n)$ time.

	We will prove that each triangle~$\{x,y,z\}$ is found in the least ancestor node~$p$ of~$x$,~$y$, and~$z$. 
	As each node in~$B$ references the triangles of its children,~$\{x,y,z\}$ is passed on to the root node in~$B$. 
	Note that each node in~$B$ only computes those new triangles which have one vertex in the subtree rooted in one child node and two vertices in the subtree rooted in the other child node. 
	Thus each triangle is computed at most once.
	Let~$q$ and~$r$ be the children of~$p$ and let without loss of generality be~\mbox{$x \in Q_q(s), y \in Q_q(t)$} and~$z \in Q_r(u)$. 
	Since~$\{x,z\},\{y,z\}\in E$, it holds that~$F_{q,s,r,u} \neq \emptyset$ and~$F_{q,t,r,u} \neq \emptyset$. 
	Moreover, it holds that~$\{x,y\}\in E^q_{s,t}$ because~$x \in Q_q(s), y \in Q_q(t)$, and~$\{x,y\}\in E$. 
	Hence,~$\{x,y,z\}$ is contained in
	\begin{align*}
	 \big\{\{a,b,c\}\mid a\in Q_q(s) &\land b\in Q_q(t) \land c\in Q_r(u) \land \\
				      &\land F_{q,s,r,u} \neq \emptyset \land F_{q,t,r,u} \neq \emptyset \land \{a,b\}\in E^q_{s,t}\big\},
	\end{align*}

%
	which is a subset of~$LT_{q,q,r}$.
	Since~$LT_{q,q,r}\subseteq LT^p$, $\{x,y,z\}$ is found in node~$p$. 
  \end{proof}

\section{Conclusion}
\label{sec:conclusion}
Employing the framework of FPT-in-P analysis \cite{GMN17}, we provided novel notions and insights concerning potentially faster algorithms for enumerating (and detecting) triangles in undirected graphs.
One the one hand, it remains to be seen whether General-Problem-hardness is an appropriate notion for intractability within the field of FPT~in~P.
On the other hand, so far there is still little work on kernelization in the context of enumeration problems;
we hope that the notion of enum-advice kernels can be used to further develop this area of research.

As previously observed by Ortmann and Brandes~\cite{OB14}, the parameterized algorithm of Chiba and Nishizeki~\cite{DBLP:journals/siamcomp/ChibaN85} despite its age is still very competitive.
Experiments revealed that their algorithm enumerated all triangles in social networks with several hundred thousand vertices within a few seconds.
Our algorithms for the parameter distance to~$d$-degenerate graphs had to perform less operations for small values of~$d$; we were, however, not able to beat the actual running time of the algorithm by Chiba and Nishizeki \cite{DBLP:journals/siamcomp/ChibaN85}.
An analysis revealed that their algorithm uses main memory of modern computers with their memory hierarchy more efficiently. In particular, our algorithm accesses data that is scattered throughout the memory.
We conclude that further practical improvement for \trenum{} should include aspects of memory-efficiency in the algorithm design process.
Independent to our work, a first step into this direction has already been done by~\citet{DBLP:journals/tods/HuTC14}.

It remains open to study whether our exponential-size kernel for the parameter ``distance to $d$-degenerate graphs'' (see \cref{thm:dtddkernel}) can be improved in terms of size and running time. 
On a more general scale, note that triangles are the smallest non-trivial cliques as well as cycles. 
Can one generalize our findings to these two different settings when increasing the subgraph size?
Finally, we mention that following the FPT-in-P route might be an attractive way to ``circumvent'' lower bound results for other polynomial-time solvable problems.
To this end, a systematic exploration of parameter spaces (cf.\ \citet{parahi} or \cref{fig:paramh}) and parameter combinations~\cite{Nie10} seems beneficial.

\paragraph{Acknowledgement}
We are grateful to Philipp Zschoche for providing the practical insights and preliminary experiments.
We thank Mark Ortmann for fruitful discussions on obstacles for practical algorithms.
Finally, we are grateful to two anonymous reviewers of \emph{Journal of Computer and System Sciences} whose constructive feedback helped to significantly improve the presentation.

\setlength{\bibsep}{0pt plus 0.1ex}
\bibliographystyle{plainnat}
\bibliography{mylib}

\end{document}